\documentclass[twocolumn]{autart}    

\usepackage{graphics}
\usepackage{graphicx}
\usepackage{amsmath}
\usepackage{amssymb}
\usepackage{epsfig,wrapfig}
\usepackage{amsmath}
\usepackage{color}
\usepackage{float}
\usepackage{amsmath}
\usepackage{amsfonts}
\usepackage{amssymb}
\usepackage{calrsfs}
\usepackage{cite}
\usepackage{verbatim}
\usepackage[title]{appendix}
\usepackage{xcolor}

\usepackage{wrapfig}

\newcommand{\Eta}{\Tilde{\eta}}
\newcommand{\R}{\mathbb{R}}
\newcommand{\BM}{\begin{bmatrix}}
\newcommand{\EM}{\end{bmatrix}}

\newcommand{\be}{\begin{equation}\begin{aligned}}
\newcommand{\ee}{\end{aligned}\end{equation}}

\newtheorem{definition}{Definition}
\newtheorem{theorem}{Theorem}

\newtheorem{corollary}{Corollary}

\newtheorem{remark}{Remark}
\newtheorem{assumption}{Assumption}
\newtheorem{example}{Example}

\begin{document}

\begin{frontmatter}

\title{Heterogeneous mixtures of dictionary functions to approximate subspace invariance in Koopman operators} 

\thanks[footnoteinfo]{ Corresponding Author at: The Biological Computing and Learning Laboratory at UC Santa Barbara, Santa Barbara, California 93106}

\author[ey]{Charles A. Johnson}\ead{cajohnson@ucsb.edu}$^{,*}$,    
\author[sh]{Shara Balakrishnan}\ead{sbalakrishnan@ucsb.edu}, 
\author[ey]{Enoch Yeung}\ead{eyeung@ucsb.edu}               

\address[ey]{Department of Mechanical Engineering, University of California, Santa Barbara, 93106, United States}     

\address[sh]{Department of Electrical and Computer Engineering, University of California, Santa Barbara, 93106, United States}


\begin{abstract}                         
Koopman operators model nonlinear dynamics as a linear dynamic system acting on a nonlinear function as the state. 
This nonstandard state is often called a Koopman observable and is usually approximated numerically by a superposition of functions drawn from a \textit{dictionary}. A widely used algorithm, is \textit{Extended Dynamic Mode Decomposition}, where the dictionary functions are drawn from a fixed, homogeneous class of functions.  Recently, deep learning combined with EDMD has been used to learn novel dictionary functions in an algorithm called deep dynamic mode decomposition (deepDMD). The learned representation both (1) accurately models and (2) scales well with the dimension of the original nonlinear system. In this paper we analyze the learned dictionaries from deepDMD and explore the theoretical basis for their strong performance. We discover a novel class of dictionary functions to approximate Koopman observables.  Error analysis of these dictionary functions show they satisfy a property of subspace approximation, which we define as uniform finite approximate closure. We discover that structured mixing of heterogeneous dictionary functions drawn from different classes of nonlinear functions achieve the same accuracy and dimensional scaling as deepDMD. This mixed dictionary does so with an order of magnitude reduction in parameters, while maintaining geometric interpretability. Our results provide a hypothesis to explain the success of deep neural networks in learning numerical approximations to Koopman operators. 
\end{abstract}

\end{frontmatter}

\section{Introduction}



Koopman operator theory considers an alternate representation of a dynamic system where the state evolution of a nonlinear system is linear.  In this representation, the concepts of vibrational, growth, and decay modes in linear systems can be directly extended to nonlinear systems \cite{mezic2005spectral}. These modes address problems in the field of fluid mechanics \cite{rowley2009spectral, schmid2010dynamic, mezic2013analysis} and disease modeling \cite{proctor2018generalizing}, as well as programming the steady state of biological systems \cite{hasnaindata} or extracting new biosensors \cite{hasnain2022learning}.  The spectral properties of the linear Koopman operator in this function space connects to model reduction, validation, identification and control \cite{mauroy2020Koopman}. Since Mezić's first paper on the spectral properties of the Koopman operator, computing Koopman modes has become a major research focus \cite{budivsic2012applied, williams2014kernel}. The central algorithm for computing these modes is dynamic mode decomposition (DMD).

DMD-based methods, such as extended dynamic mode decomposition (EDMD) \cite{williams2015data} give an attractive set of tools to model nonlinear dynamics from data. In EDMD, one chooses a nonlinear function space implicitly defined by the span of a predefined, homogeneous function dictionary and then, individual modes are learned to compute a low-rank approximation to the linear operator. Because a human typically specifies the dictionary functions in EDMD, the resulting models tend to be high dimensional. One approach to learning lower dimensional models is through the SINDy algorithm. This algorithm uses sparse regression to project to a lower dimensional subspace of the nonlinear function space initially chosen \cite{brunton2016discovering}. 

A second approach we have explored is to use deep learning to learn a function dictionary or a set of observables during training \cite{yeung2019learning}.  deepDMD uses a form of stochastic gradient descent (SGD) to train a deep artificial neural network to select observables.  Because the neural network is a universal function approximator \cite{hornik1991approximation}, it can learn any set of observables and parameterize the approximate Koopman operator. The deepDMD algorithm has been generally successful at learning both small and large scale systems with many types of strong nonlinearities \cite{yeung2019learning}.

Other approaches leverage deep learning to build efficiently parameterized Koopman models with high fidelity.  Methods using linear recurrent autoencoders build effective, low-dimensional Koopman models that are not measurement-inclusive \cite{wehmeyer2018time, lusch2018deep, otto2019linearly}. To do this, they pass observables from the learned function space, via the decoder, back to the original model's state-space. Other work has, under the assumption of ergodicity, delegated the process of choosing a lifting map for EDMD to deep learning \cite{takeishi2017learning}. 


Neural network models are typically black box. Simple parametric representations of the functions that neural networks learn are elusive. We use dictionaries of functions inspired from deep learning to build Koopman models which both capture the system dynamics, and approximate the dynamics of the nonlinear dictionary functions. The success of these models provides insight into why deep learning models have the capacity to approximate Koopman operators well.

We also follow the pattern we see from deep learning by mixing classes of dictionary functions to build low-dimensional models. The mixed function bases we study give further insight into the strategy adopted by deepDMD to solve the data-driven Koopman learning problem. Combining our new function dictionaries with an effective learning algorithm yields a training loss comparable to deepDMD, but with an order of magnitude reduction in model complexity.  

\section{The Koopman Generator Learning Problem}\label{sec:KLP}
In this section we introduce the problem of choosing a function space and approximate Koopman generator to model a dynamic system.  This problem is not convex, but, when approximated well, it gives an accurate, data-driven, linear model of a nonlinear system. 

Consider a nonlinear, time-invariant, autonomous system with dynamics
\begin{equation}\label{eq:nonlinear_system}
\dot{x} = f(x)
\end{equation}
where $x \in M \subset \mathbb{R}^n$, $f:\mathbb{R}^n \rightarrow \mathbb{R}^n$ is analytic.  The manifold, $M$, is the state-space of the dynamical system.  
We introduce the concepts of a Koopman generator and its associated multi-variate Koopman semigroup, following the exposition of \cite{budivsic2012applied}. 


For continuous nonlinear systems, the Koopman semigroup is a semigroup, a set with an associative binary operation, $\mathcal{K}_{t\in \mathbb{R}}$ of linear but infinite dimensional operators, $\mathcal{K}_t$, that acts on a space of functions, $\Psi$, with elements $y: M \rightarrow \mathbb{R}^m$.
We assume that each $y\in\Psi$ is differentiable with a bounded derivative. Our function, $y$, is an observable because it is a function of the state, $x$.  We say $\mathcal{K}_t:\Psi \rightarrow \Psi$ is an operator for each $t\geq 0$.  
The Koopman operator applies the transformation, 
\begin{equation}\label{eq:koopman_def} \mathcal{K}_t \circ y(x) = y \circ \Phi_t(x),\end{equation}
where $\Phi_t(x)$ is the flow map of the dynamical system (\ref{eq:nonlinear_system}) evolved forward up to time $t$, given the initial state $x$. Instead of examining the evolution of the state, the Koopman semigroup allows us to study the forward evolution of functions of state, $y(x)$ \cite{williams2015data}. 

The generator, $\mathcal{K_G}$, for the Koopman semigroup is defined as 
\begin{equation}
\mathcal{K_G} \circ y  \triangleq \lim_{t\rightarrow 0} \frac{ \mathcal{K}_t \circ y - y }{t}.
\end{equation}

When $y$ and $t$ are fixed, we see from Eq. (\ref{eq:koopman_def}) that $\mathcal{K}_t$ is a function of the state, $x$. Similarly, $\mathcal{K_G}$ may be understood as a function of $x$.  The Koopman generator is a state-dependant derivative operator, see Section 7.6 of \cite{lasota1998chaos}. It satisfies
\be\label{eq:dyn_sys_in_data}
\frac{d}{dt}y(x)=\mathcal{K_G}(x)\circ y(x).
\ee In rare cases the Koopman generator will not have a dependence on $x$, see \cite{brunton2016koopman}.

\subsection{Learning Koopman operators from data}
In discrete-time, data-driven Koopman operator learning, we have $r$  pairs of measurements of observables, \[(y(x_i), y(f(x_i))),\mbox{ for }i=1,2,...,r.\] In continuous-time (CT), data-driven Koopman generator learning, our pairs are measurements of observables and their derivatives \[(y(x_i), d(y(x_i))/dt),\mbox{ for }i=1,2,...,r.\] 

Because we are working with numerical computation, these measurements and their derivatives are finite dimensional, call the dimension $m<\infty$. We assume that our measurements are higher dimensional than the underlying dynamic system, $m\geq n$, and that $y(x)$ is injective. We assume all measurements are noise-free.

The Koopman generator (for CT) $\mathcal{K_G}(x)$ is unknown, however, we assume it exists and satisfies \[ \frac{d(y(x_i))}{dt} = \mathcal{K_G}(x_i)\circ y(x_i)\mbox{ for }i=1,2,...,r.\]

While we do not know the values of $x$ from our data we can write down our derivative function as an implicit function of $x$, this function is, $F:\R^m\rightarrow \R^m$, and it satisfies \[F(y) = \mathcal{K_G}(h^{-1}(y))\circ y,\] where $y=h(x)$. The function $h$ is invertible as we assume that $y$ is injective.  

The implicit function, $F$, simply is Eq. (\ref{eq:dyn_sys_in_data}) rewritten in terms of $y(x)$ instead of $x$. This is an important distinction as the true system state $x$ is not necessarily known. We are working from measurements (observables) of the system in Eq. (\ref{eq:nonlinear_system}). We label these measurements $y$.   We then, choose our own finite set of $N$ {\em dictionary functions}, $\psi(y):\R^m\rightarrow \R^N$, and a constant matrix, $K\in\R^{N\times N}$ to model $F$. 

Note that we choose $\psi$ without knowing the equation for $y(x)$. The function $y(x)$ is an observable function, as it is a measured function of the state of the system given in Eq. (\ref{eq:nonlinear_system}). Likewise, as $\psi$ is a function of an observable, $\psi$ is a function of $x$ as well. Therefore, in a data-driven setting, dictionary functions of $y(x)$ are also observables of some Koopman generator. Typically, we choose $N$ such that $N>m$, and so $\psi(y)$ is a ``lifting'' of our data (see the use of ``lifting'' in \cite{korda2018convergence}).

In numeric methods for Koopman modeling, we approximate $F$ with the matrix-vector product, $K\psi(y)$.  Effectively, in a data-driven setting, this amount to projecting the action of an infinite dimensional Koopman operator as matrix multiplication on sampled data space. For example, when $N>m$, we can define a projection function $\mathtt{P}:\R^N\rightarrow \R^m$ that maps the matrix-vector product of the approximate Koopman operator $K$  multiplying the dictionary function $\psi(y)$  to the vector field $F(y)$,  such that
\be \mathtt{P}\circ K \circ \psi \circ y \approxeq  F \circ y .\ee 
This, in summary, yields the following finite approximation to the Koopman generator equation
\be\frac{dy(x)}{dt}=\mathcal{K_G}(x)y(x) \triangleq F(y(x)) \approxeq \mathtt{P}(K\psi(y)).\ee 

The matrix, $K$, is a linear operator, to ensure that it behaves like a true Koopman generator we choose it, in conjunction with $\psi$, to satisfy \be\label{eq:koop_approx_behavior}\frac{d\psi(y)}{dt}\approx K\psi(y).\ee  

We are tasked to learn a numerical, finite-dimensional approximation $K$ of the Koopman generator, $\mathcal{K_G}$ and the set of dictionary functions, $\psi$.  This matrix $K$ is the Koopman generator approximation that acts specifically on data-centered evaluations of dictionary functions $\psi(y)$ rather than the observable function $y(x).$

\subsection{Problem Statement}
Finding the Koopman generator $\mathcal{K_G}$ from data is  impossible, due to the lack of knowledge about the true Koopman generator and the parametric form of the measurements.  So, we instead solve the following optimization problem in its place.
\begin{equation}\label{eq:objective}
\min_{K, \psi \in \Psi} \sum_{i=1}^{r} \left\Vert  \frac{d\psi(y(x_i))}{dt}  - K \psi(y(x_i)) \right\Vert.
\end{equation}

In this optimization we need to select dictionary functions, as well as a real-valued matrix, $K$. This optimization problem is non-convex, since the form of $\psi(y)$ is unknown or parametrically undefined.  The model dimension, $N$, is a hyperparameter of Eq. (\ref{eq:objective}).

In EDMD the dictionary functions $\psi(y)$ are predefined, drawn from a  {\it homogeneous} class of nonlinear functions, and assumed to be known.  Under these assumptions, Eq. (\ref{eq:objective}) is an affine optimization problem with a closed-form solution. 

By contrast, in deepDMD, the dictionary functions and $K$ are learned simultaneously during iterative training.  This is, of course, a nonlinear, non-convex optimization problem, for which we employ variants of the stochastic gradient descent algorithm from Tensorflow or Pytorch, such as adaptive gradient descent (AdaGrad) or adaptive momentum (ADAM).   These numerical approaches provide no guarantee that the learned set of dictionary functions $\psi_1(y), \psi_2(y), ...,\psi_{N}(y)$ are homogeneous, or drawn from the same function class. The outcomes of deep dynamic mode decomposition often produce heterogeneous dictionaries $\psi(y)$.


\begin{figure*}[ht]
    \centering
    \includegraphics[width=\linewidth]{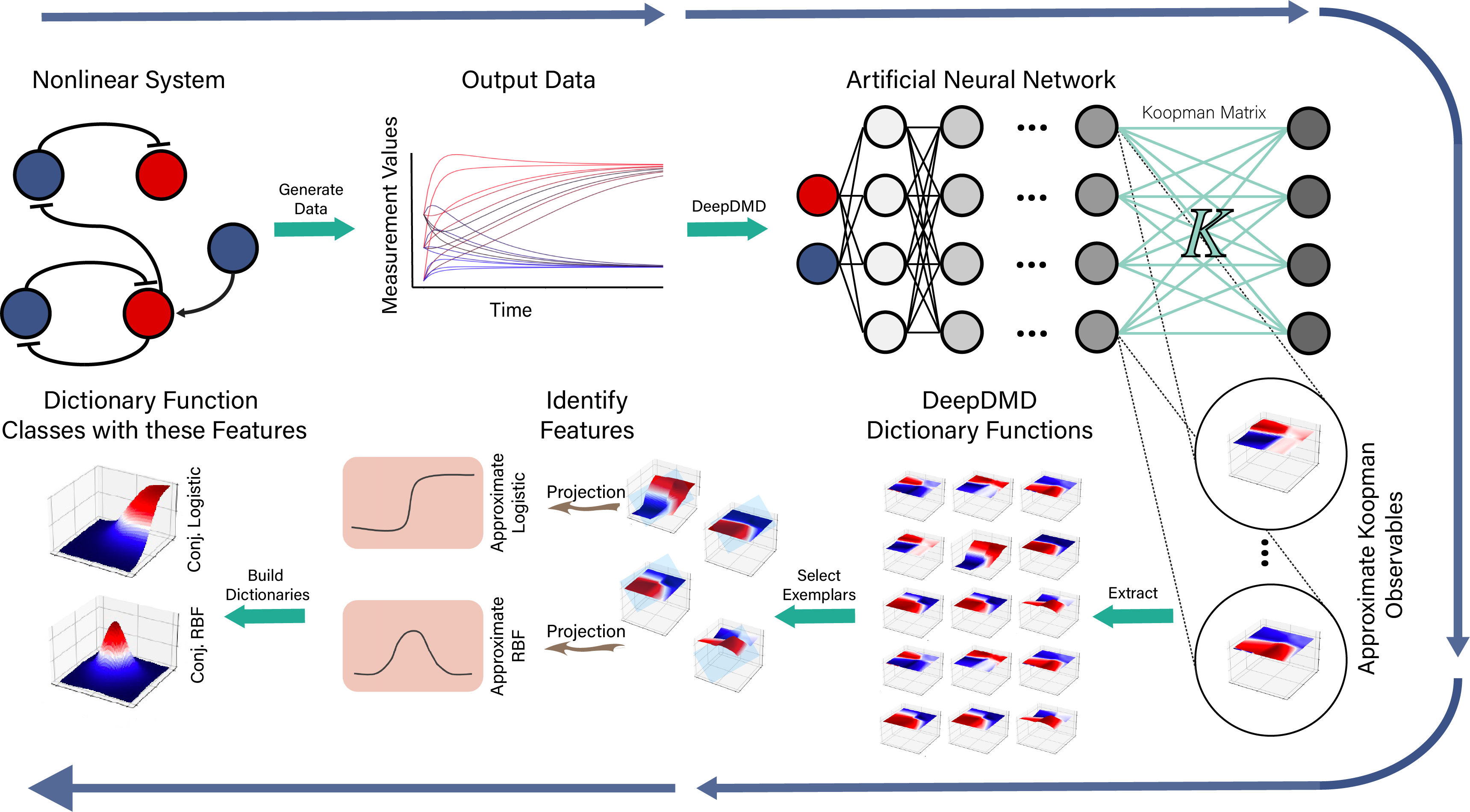}
    \caption{Process for the discovery of a novel mixed function dictionary with approximate subspace invariance. First, the deepDMD algorithm takes data from a nonlinear system to build approximate Koopman observables. Then, projections of these observables are analyzed for these functional properties. Finally, high-dimensional functions, whose projections satisfy these properties, are developed and their closure properties are verified.}
    \label{fig:DDMD_0bsTogg}
\end{figure*}

\subsection{Finite Closure}\label{sec:error}

Previous work characterizes the closure and convergence of Koopman models as additional dictionary functions are appended to the model for the DMD and EDMD algorithms \cite{arbabi2017ergodic, korda2018convergence}. We explore closure as an inherent property of a dictionary. We begin by understanding the property of subspace invariance, which corresponds to finite exact closure.

Let $S$, be the span of our dictionary functions. The set $S$ is a subspace of the set of all analytic functions, including $y$, that map $x$ to $\R$.

\begin{definition}
We say that our dictionary, $\psi(y)$, satisfies Koopman \textit{subspace invariance} when $\mathcal{K_G}\cdot \psi(y)\in S$. We call a Koopman subspace invariant dictionary a dictionary that satisfies \textit{finite exact closure}.
\end{definition}

If some element in a dictionary does not satisfy Koopman subspace invariance, then there exists some element in the dictionary that, when acted on by the Koopman operator, cannot be represented as a linear combination of dictionary functions.  In the context of data-driven Koopman learning this means that when our dictionary contains $N$ functions, no $N$ by $N$ matrix captures the precise action of the Koopman generator.  

Finite exact closure (Koopman subspace invariance) is unlikely to be achieved, as (1) the Koopman generator may need to be infinite dimensional and (2) in the case that it does not, it is difficult to engineer models with exact closure even with an explicit knowledge of Eq. (\ref{eq:nonlinear_system}).  So, we also consider an approximate notion of closure relevant to building models from data. 
\begin{definition}\label{def:uniform} We say $\psi(y):\R^m \rightarrow \mathbb{R}^{N}$ achieves finite $\epsilon$-closure or finite closure with $O(\epsilon)$ error when there exists a $K\in \mathbb{R}^{N\times N}$ and an $\epsilon > 0$ such that 
\begin{equation}
\frac{d(\psi(y))}{dt} = K \psi(y) + \epsilon(y),
\end{equation} for the vector field $F$.

Our dictionary $\psi(y)$ achieves \textit{finite approximate closure} when, for the vector field, $F$, and every $y$, the function $\epsilon(y)$ is a bounded for every $K$ such that $||K||<\infty$.

We say that $\psi(y)$ achieves {\em uniform finite approximate closure} for some set $\mathcal{R}$ when it achieves finite closure with $|\epsilon(y)| < B \in \mathbb{R}$ for all $y \in \mathcal{R}.$
\end{definition}

We are most interested by the property of uniform finite approximate closure, especially when we can bound the constant, $B$, to be arbitrarily low as this corresponds to an accurate, data-driven model.

\begin{example}\label{example}
Unfortunately, closure does not come with every dictionary of observables, even if that dictionary spans the function space of the dynamic system. For example, a canonical polynomial basis, $\{1, y, y^2,...,y^n\}$, used to approximate the one-dimensional system $f(y)=y^2$, spans the dynamic system, but no model using this basis will be closed. In fact, none will achieve uniform finite approximate closure.  We illustrate what this lack of closure means when $y>>1$. In that case, the Lie derivative of $y^n$ is  \[\frac{d(y^n)}{dt} = \frac{d(y^n)}{dy}\frac{dy}{dt} = ny^{n-1}y^2  = ny^{n+1}.\] Approximating this derivative when $y>>1$ as an $n^{th}$ degree polynomial will dramatically fail as the error will be of order $O(y^{n+1})$.  This failure will cascade back through approximations of all the other dictionary functions. Closure is a crucial property of these models! 
\end{example}

We want models with uniform finite approximate closure because, as the bounding constant goes to zero, $B \rightarrow 0$, we may use $K$ to perform stability, observability and spectral analysis. We see this is true as, in the discrete-time formulation, as $B\rightarrow 0$, $K$ approaches a projection of the action of the Koopman operator \cite{korda2018convergence}. 
When our dictionary is state-inclusive, it is trivial to project from from $\psi$ to $y$ and its trajectory may yield stability insights. To keep our models meaningful in this way, all the dictionaries in this article are state-inclusive.

\section{SILL: A Homogeneous Dictionary Model of deepDMD's Learned Dictionary}\label{sec:SILL}
In this section, we define a new class of dictionary functions identified from deepDMD's solution to Eq. (\ref{eq:objective}). We will call this class \textit{State-Inclusive Logistic Liftings} (SILLs). We call them this because
\begin{enumerate}
    \item they contain the state of the vector field $F(y)$, the dynamics of the governing equations (as we have measured them) are directly included in the dictionary, so they are state-inclusive,
    \item the dictionaries in this class also contain nonlinear functions, all of which are conjunctive logistic functions, so these dictionaries are logistic in nature, and
    \item because each has at least one nonlinear dictionary function so the dictionary size, $1+m+N$, is greater than the number of measurements, $m$. Since $1+m+N>m$, the model using this dictionary is ``lifted'' to a higher dimension than the original measurements.
\end{enumerate}   

We previously showed that Koopman models chosen from SILL dictionaries have successfully learned global nonlinear phase-space behavior of several simple, nonlinear systems \cite{johnson2018class}.  In Section \ref{sec:SILLclosure} we show for the first time that SILL dictionaries satisfy uniform finite approximate closure. 

Note that the functions that we define technically are measurement-lifting, not state-lifting. This is a consequence of how we cast our problem formulation.

The notation in the sections that follow in this paper is explained in Section \ref{sec:AppendixNotation} of the Appendix.


\subsection{The SILL Lifting Functions and deepDMD}
We define a multivariate conjunctive logistic function, $\Lambda:\R^m\rightarrow\R$, for a given center parameter vector $\mu_j \in \mathbb{R}^m$ and steepness parameter vector $\alpha_j \in \mathbb{R}^m$   as follows
\begin{equation}
\Lambda(y; \mu_j, \alpha_j) \triangleq \prod_{i=1}^{m}\lambda(y_i; \theta_{ji}),
\end{equation}
where $y \in \mathbb{R}^m$ and the scalar logistic function, $\lambda:\R\rightarrow\R$, is defined as 
\begin{equation}\label{logistic}
\lambda(y_i; \mu_{ji}, \alpha_{ji}) \triangleq \frac{1}{1+e^{-\alpha_{ji} (y_i-\mu_{ji})}}.
\end{equation}
The parameters $\mu_{ji}$ define the centers or the point of activation for $\Lambda(x;\theta_j)$ along dimension $y_i$.  The parameter $\alpha_{ji}$ is a steepness or sensitivity parameter, and determines the steepness of the logistic curve in the $i^{th}$ dimension for $\Lambda(y;\theta_j)$. Conjunctive logistic functions map orthants of $R^m$ to be nearly 1 and the rest of the space to be nearly 0.

To illustrate how this conjunctive logistic function works, consider what happens if you set the vector $y$ to be constant in all but the $l^{th}$ dimension. Then $\Lambda(y;\mu_j,\alpha_j) = (\prod_{i\neq l}c_i)\lambda(y_l;\mu_{jl}, \alpha_{jl})$. This is a constant times the logistic function in the $l^{th}$ dimension. When we project dictionary functions learned by deepDMD to a single dimension we observe that they likewise approximate a scaled logistic function.

Given $N$ multivariate logistic functions, we  define a SILL dictionary as $\psi : \R^m\rightarrow \R^{1+m+N}$ so that:
\begin{equation}
\psi(y) \triangleq \begin{bmatrix}
1&
y^T&
\bar\Lambda(y)^T
\end{bmatrix}^T
\end{equation}
where $\bar\Lambda(y) = [\Lambda(y;\theta_1), \hdots,\Lambda(y;\theta_N) ]^T$ is a vector of conjunctive logistic functions. We then have that $K\in\R^{(1+m+N)\times (1+m+N)}$.   This basis is measurement-inclusive. Ideally $\bar\Lambda(y)$, the measurements themselves and a constant spans each dimension of the vector field over the region of interest.


\section{AugSILL: A Heterogeneous Dictionary Model of DeepDMD's Learned Dictionary}\label{sec:augSILL}


Projections of deepDMD's dictionary functions tended to match the profiles of logistic functions, however there were some punctuated irregularities in these projections (See Fig. \ref{fig:DDMD_0bsTogg}).  We choose to model these irregularities with radial basis functions (RBFs).  We therefore augment our SILL dictionary with conjunctive multivariate RBFs (conjunctive RBFs).  We define an RBF with a steepness of $\alpha_{ki}$ and a center of $\mu_{ki}$ as follows: \[\rho(y_i; \mu_{ki}, \alpha_{ki}) \triangleq \frac{e^{-\alpha_{ki} (y_i-\mu_{ki})}}{(1+e^{-\alpha_{ki}(y_i-\mu_{ki})})^2}.\] 

This RBF takes on its global maximum value of $\frac{1}{4}$ when the measurement $y_i=\mu_{ki}$ the center parameter. It radially approaches zero at a rate determined by the value of the steepness parameter $\alpha_{ki}$.

Note the following relationship between our RBF and logistic functions. The RBF  $\rho(y_i;\theta_{ki}) = \lambda(y_i;\theta_{ki}) - \lambda(y_i;\theta_{ki})^2$. So, even in one dimension, exactly approximating an RBF would require an infinite linear combination of SILL functions (a piecewise linear spline on infinitesimally small intervals). Thus, conjunctive RBFs contribute  distinct nonlinear features that are outside the span of the SILL function space.  This mathematical observation is the rationale for referring to this mixture of dictionary functions as heterogeneous.

We define an $m$-dimensional conjunctive RBF to be: \[P(y;\mu_k, \alpha_k)\triangleq\prod_{i=1}^m \rho(y_i;\theta_{ki}).\]  

Conjunctive RBFs map their center to a value of $4^{-m}$ and radially around that center drop off to zero. The steepness parameters determine how quickly the drop off to zero occurs along each coordinate axis.

When you set the vector $y$ to be constant in all but the $l^{th}$ dimension, $P(y;\mu_j,\alpha_j) = (\prod_{i\neq l}c_i)\rho(y_l;\mu_{jl}, \alpha_{jl})$. This is a constant times the RBF  in the $l^{th}$ dimension. When we project dictionary functions learned by deepDMD to a single dimension, we observe that in addition to approximating scaled logistic functions, many approximate scaled RBFs and scaled sums of RBFs and logistic functions. 

We propose the augSILL dictionary as the stacked vector of a mixture of dictionary functions of $y$ \[ \psi(y) \triangleq \begin{bmatrix}
1&y^T&\bar \Lambda^T&\bar P^T
\end{bmatrix}^T, \] where the vector $\bar\Lambda \triangleq [\Lambda(y;\theta_1), ..., \Lambda(y;\theta_{N_L})]^T$  contains all conjunctive logistic functions, and $\bar P \triangleq [P(y;\theta_{N_L+1}), ..., P(y;\theta_{N_L+N_R})]^T$ contains all conjunctive radial basis functions. Here, $N_L$ and $N_R$ are non-negative integers. In this article, the \textit{augSILL} (augmented SILL) basis includes $N_L$ conjunctive logistic functions and $N_R$ conjunctive RBFs. To our knowledge this is the first time that anyone has analyzed the behavior of mixed dictionary functions for the numerical approximation of a Koopman operator or Koopman generator.

Can a mixture of function basis form a coherent basis that preserves subspace invariance, or at least uniform finite approximate closure? Certainly the more varied basis facilitates approximating the vector field, $F(y)$. But the subspace invariance properties are just as important to Koopman models as approximating the vector field (see Example \ref{example}).

\section{Uniform Finite Approximate Closure of the SILL Dictionary}\label{sec:SILLclosure}
We showed in \cite{johnson2018class} that a pure, measurement-inclusive, SILL dictionary can make effective low-dimensional Koopman models. Between this section and Appendix \ref{sec:averageErrorSILL} we demonstrate that the homogeneous SILL dictionary satisfies uniform finite approximate closure (see Definition \ref{def:uniform} in Section \ref{sec:error}).  In essence, uniform finite approximate closure guarantees that the dimensionality of the dictionary space does not need to diverge to infinity, while simultaneously approximating the vector field $F(y)$ sufficiently well.  This property is what makes numerical approximation of the Koopman generator equation possible. 

Recall from Example \ref{example} that when considering a new set of dictionary functions, we also have to consider the effects of dictionary explosion.  In Example \ref{example}, we computed the Lie derivatives of each dictionary function, which in turn generated new product terms that were not in the span of the existing dictionary.  Thus, for any new class of dictionary functions, a key property requisite to uniform approximate finite closure is the ability to approximate products of dictionary functions as elements of the span of the dictionary.


The first step is to show convergence in steepness of products of conjunctive logistic functions to a single conjunctive logistic function. We hypothesize that the product of two conjunctive logistic functions may be approximated as a single conjunctive radial basis function as follows (see Fig. \ref{fig:rbftsig}\textbf{a}):
\begin{equation}\label{eq:approx1}
    \Lambda(y;\theta_l)\Lambda(y;\theta_j) \approx \Lambda(y;\theta^*)
\end{equation}
where $\theta^* = (\mu_{max}(l,j), \alpha_{max}(l,j))$ and 
\begin{equation} 
\mu_{max}(l,j) = \left( \max\{\mu_{l1},\mu_{j1}\}, ..., \max\{\mu_{lm},\mu_{jm}\} \right),
\end{equation} and $\alpha_{max}(l,j)$ are the $\alpha$ values that correspond to the indices of the $\mu$'s.

Theorem \ref{thm:SILLconv} demonstrates that the approximation in Eq. (\ref{eq:approx1}) is a good approximation in the limit of increasing steepness parameter $\alpha$.  The proofs of all theorems and corollaries that are not included in the body this article are included in Section \ref{sec:AppendixProofs} of the Appendix.

\begin{theorem}\label{thm:SILLconv}
Under Assumption \ref{assump:order}, if the dictionary functions do not exactly match their corresponding center parameters, $y_i\neq\mu_{ji} $ for all $ i\in \{1, 2, ..., m\}$ and  $j\in\{1,2,...,N\}$, then, as the steepness parameters go to infinity, the product of two conjunctive logistic function will exponentially approach a single conjunctive logistic function in the dictionary, $\alpha\rightarrow\infty$, \[\Lambda(x;\theta_l)  \Lambda(x;\theta_j) - \Lambda(x;\theta^*)\rightarrow 0\] exponentially.
\end{theorem}

\begin{remark}
Given a finite SILL dictionary that does not satisfy Assumption \ref{assump:order}, one can enforce that Assumption \ref{assump:order} holds by adding a finite number of additional conjunctive logistic functions to the basis.
\end{remark}

Theorem \ref{thm:SILLconv} implies an intermediary result to demonstrating uniform finite approximate closure of the SILL basis. This result connects Theorem \ref{thm:SILLconv} to the Lie derivative of a SILL dictionary function. 

\begin{corollary}\label{cor:logSILLApprox}
Under the assumptions of Theorem \ref{thm:SILLconv}, when $F$ is spanned by a SILL dictionary, the Lie derivative of a conjunctive logistic function exponentially approaches a finite weighted sum of conjunctive logistic functions as the steepnesses of the functions goes to infinity. Specifically, \be\label{eq:nonlinearCombforSILL}\dot\Lambda(x;\theta_l)\rightarrow \sum_{i=1}^{n}\sum_{j=1}^{N}\alpha_{li}w_{ij}(1-\lambda(x_i;\theta_{li}))\Lambda(x;\theta^*)\ee exponentially as $\alpha\rightarrow\infty$.
\end{corollary}

We now show the second step in an alternate path to characterizing the finite closure properties of the SILL dictionary. This is to approximate the SILL observables' Lie derivative as linear combination of SILL basis functions, \be\label{eq:SILL_finalApprox}\sum_{i=1}^{m}\sum_{j=1}^{N}\alpha_{li}w_{ij}\Lambda(y;\theta^*). \ee This step corresponds to the lower right arrow in the right side of Fig. \ref{fig:paper_summary}.

\begin{corollary}
Under the assumptions of Theorem \ref{thm:SILLconv}, the error between
\be\label{eq:SILL_linearLie} 
\sum_{i=1}^{m}\sum_{j=1}^{N}\alpha_{li}w_{ij}\Lambda(y;\theta_l)\Lambda(y;\theta_j)
\ee and Eq. (\ref{eq:SILL_finalApprox}) goes to zero exponentially as $\alpha\rightarrow \infty$.
\end{corollary}

\begin{proof}
The proof follows from a direct application of Theorem \ref{thm:SILLconv}.
$\blacksquare$\end{proof}

The end result of Eq. (\ref{eq:nonlinearCombforSILL}) is a nonlinear combination of dictionary functions and Eq. (\ref{eq:SILL_linearLie}) is an approximation of the Lie derivative, $\dot\Lambda$.  We explore how nearly linear this nonlinear combination is and the approximation quality in Appendix \ref{sec:averageErrorSILL}. Appendix \ref{sec:averageErrorSILL} also ties these results to conclude that the SILL dictionary satisfies uniform finite approximate closure and that the bounding constant goes to zero, $B\rightarrow 0$, exponentially as steepness of the dictionary functions increases and the number of measurements increase. Thus, SILL dictionary functions define a spanning set for Koopman observables.

\begin{figure*}[ht]
    \centering
    \includegraphics[width=450pt]{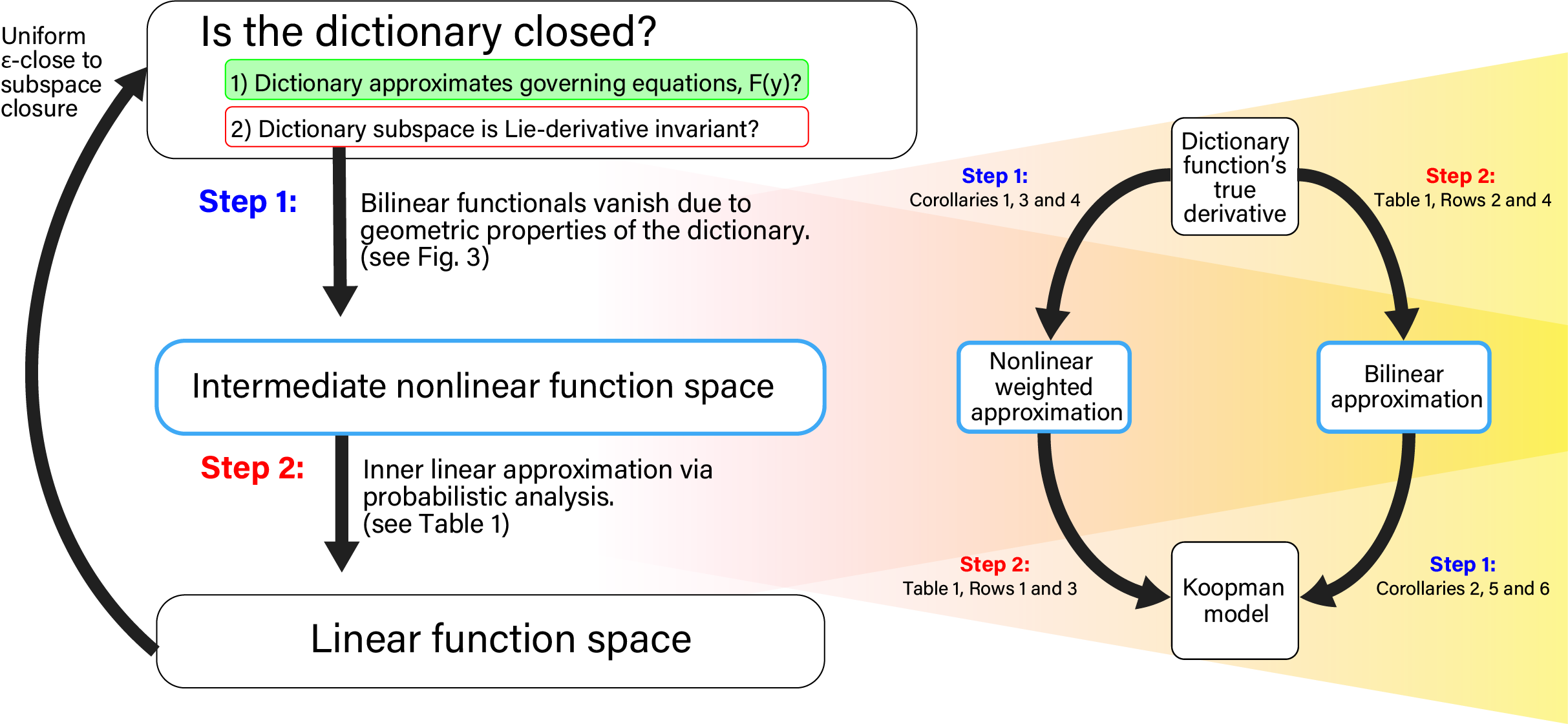}
    \caption{A visual outline of the (commutative) steps taken to prove uniform finite approximate closure of SILL dictionaries in Section \ref{sec:SILLclosure} and AugSILL dictionaries in Section \ref{sec:AugSILLclosure}. Full proofs for each of the theorems and corollaries are provided in Appendix \ref{sec:AppendixProofs}.}
    \label{fig:paper_summary}
\end{figure*}

\begin{figure*}[ht]
\centering
\includegraphics[width=450pt]{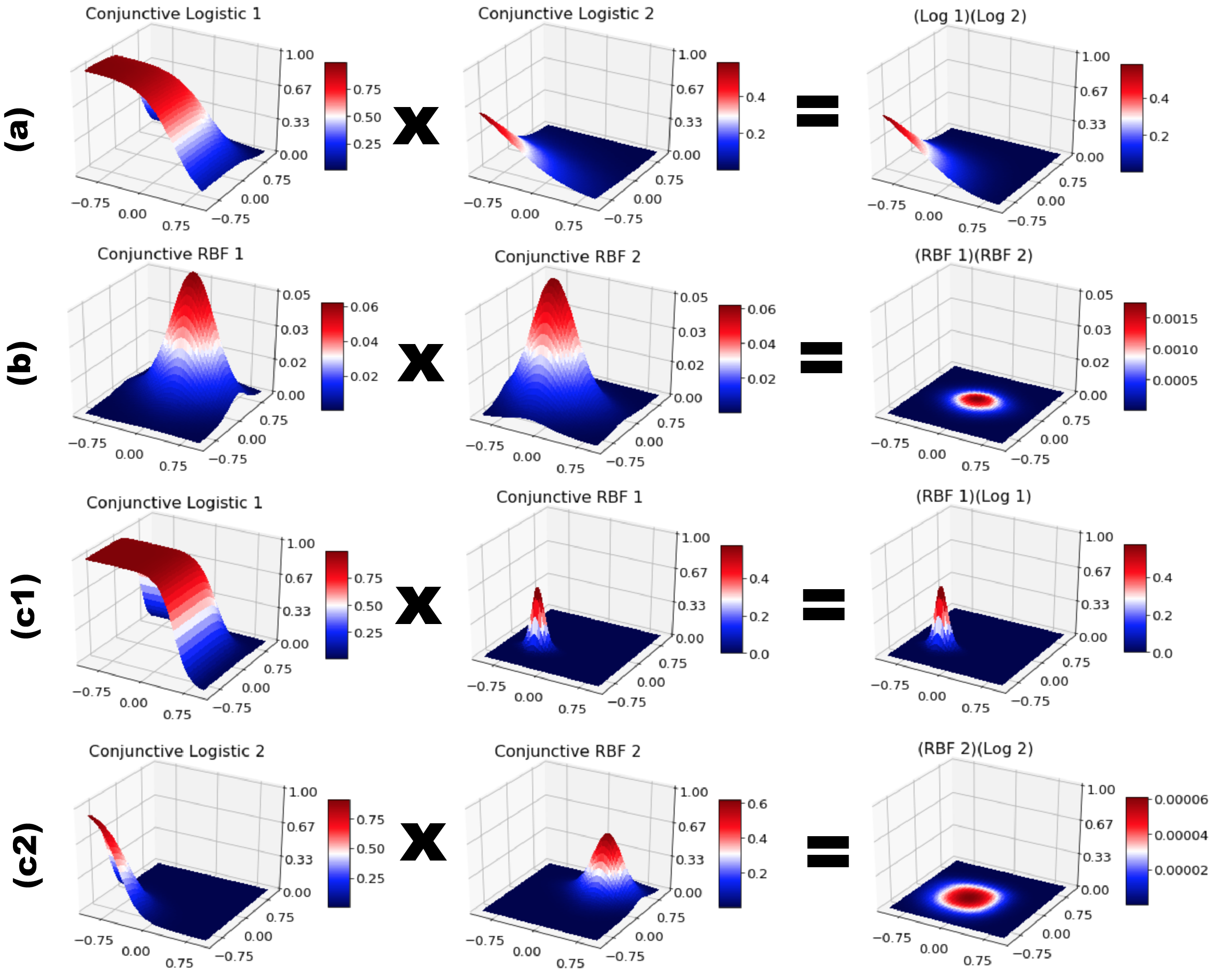}
\caption{ Visual representation of Theorems \ref{thm:SILLconv}, \ref{thm:augSILLconv1}, \ref{thm:augSILLconv2} and \ref{thm:augSILLconv4}.
\textbf{(a)} Theorem \ref{thm:SILLconv}: The product of two conjunctive logistic functions approximates the conjunctive logistic function whose centers are greater in each dimension.
\textbf{(b)} Theorem \ref{thm:augSILLconv4}: The product of two conjunctive RBFs is nearly zero unless the norm of the difference between their centers is very small.
\textbf{(c1)} Theorem \ref{thm:augSILLconv1}: The product of a conjunctive logistic function and a conjunctive RBF approximates the RBF if at least one dimension of the center in the logistic function is greater than its partner center in the RBF.  \textbf{(c2)} Theorem \ref{thm:augSILLconv2}: When this is not the case the product is nearly zero.
}
\label{fig:rbftsig}
\end{figure*}

\section{Uniform Finite Approximate Closure of the AugSILL Dictionary}\label{sec:AugSILLclosure}

In this section we analyze the subspace invariance properties of a mixed dictionary, the augSILL dictionary. To do so, we  characterize the error term in Eq. (\ref{eq:overallError}). In this section we show that the augSILL dictionary satisfies uniform finite approximate closure.


In  Sections \ref{sec:convThms1} and \ref{sec:convThms3} we establish theorems analogous to Theorem \ref{thm:SILLconv} and its corollaries for the augSILL basis.  In Section \ref{sec:augSILL_error} we demonstrate probabilistic results that combine with the corollaries in Section \ref{sec:convThms3} to uniformly bind average error of Eq. (\ref{eq:overallError}) with a constant, $B$, that can be arbitrarily small.  This approach applies to the SILL and augSILL dictionaries to show uniform finite approximate closure. 

\subsection{Convergence in Steepness of Bilinear AugSILL Terms to the Span of the AugSILL Dictionary}\label{sec:convThms1}
When we construct a mixed dictionary and compute their Lie derivatives we find mixed bilinear terms involving two types of  dictionary functions. In Section \ref{sec:SILLclosure} we demonstrated that the product of two conjunctive logistic functions can be well approximated by a single conjunctive logistic function. Here we approximate the product of a conjunctive logistic and RBF. For this approximation, (see Fig. \ref{fig:rbftsig}\textbf{c1-c2}) we have that: \begin{equation}\label{eq:approx2}
    \Lambda(y;\theta_l)P(y;\theta_k) \approx H(y;\theta_l, \theta_k)
    \end{equation} where 
\begin{equation}
H(y;\theta_l, \theta_k) \triangleq \begin{cases}
P(y;\theta_k)\mbox{ if }  \theta_l \lesssim \theta_k \\
0 \mbox{ otherwise.}
\end{cases}
\end{equation} The function, $H$, uses the relative centers of $\Lambda$ and $P$ to choose to take on the value of the conjunctive RBF, $P$, or zero. Both functions are in our dictionary as the zero function is trivially available.

In these theorems, we assume that the measurements do not exactly match up with the function centers.  This assumption is reasonable as such an exact line-up in the state space is unlikely.

In the case where the center of $\Lambda(y;\theta_l)$ is less than the center of $P(y;\theta_k)$ in at least one dimension we have that:
\begin{equation}\label{eq:P_errorterm1}
\begin{aligned}
\Lambda(y;\theta_l) & P(y;\theta_k) -  H(y;\theta_l,\theta_k) \\ 
&= (\Lambda(y;\theta_l) - 1)P(y;\theta_k). \\
\end{aligned}
\end{equation}  

Theorem \ref{thm:augSILLconv1} demonstrates that Eq. (\ref{eq:P_errorterm1}) goes to zero exponentially in the limit of the steepness parameters, $\alpha$.

\begin{theorem}\label{thm:augSILLconv1}
When the measurements do not exactly align with the centers of the dictionary functions in any dimension and the center of the conjunctive RBF is more positive than the center of the conjunctive logistic function in some dimension, then their product exponentially converges to the conjunctive RBF as their steepness parameters go to infinity. Specifically, If $y_i\neq\mu_{ki} $ for all $ i\in \{1, 2, ..., m\}$ and $k\in\{N_L+1,N_L+2,...,N\}$  and there exists $i^*\in\{1, 2, ..., m\}$ so that $\mu_{ki^*} \geq \mu_{li^*}$, then as $\alpha\rightarrow\infty$, $\Lambda(y;\theta_l) P(y;\theta_k) -  H(y;\theta_l,\theta_k)\rightarrow 0$ exponentially.
\end{theorem}

When the center of $\Lambda(y;\theta_l)$ is greater than the center of $P(y;\theta_k)$ in all dimensions we have that:
\begin{equation}\label{eq:P_errorterm2}
\begin{aligned}
\Lambda(y;\theta_l) & P(y;\theta_k) -  H(y;\theta_l, \theta_k) \\ 
&= \Lambda(y;\theta_l)P(y;\theta_k) \\
\end{aligned}
\end{equation}

Theorem \ref{thm:augSILLconv2} demonstrates that the approximation in Eq. (\ref{eq:P_errorterm2}) is a good approximation in the limit of $\alpha$.  

\begin{theorem}\label{thm:augSILLconv2}
When the measurements do not exactly align with the centers of the dictionary functions in any dimension and the center of the conjunctive logistic function is more positive than the center of the RBF in all dimensions, then their product exponentially converges to zero as their steepness parameters go to infinity. Specifically, if $y_i\neq\mu_{ki}$ and $y_i\neq\mu_{li} $ for all $ i\in \{1, 2, ..., m\}$ and for all $k\in\{N_L+1, N_L+2, ..., N\}$ so that $\mu_{ki} < \mu_{li}$, then as $\alpha\rightarrow\infty$, \[\Lambda(y;\theta_l) P(y;\theta_k) -  H(y;\theta_l, \theta_k) \rightarrow 0\] exponentially.
\end{theorem}

Between Theorems \ref{thm:augSILLconv1} and \ref{thm:augSILLconv2} we have that $\Lambda P \approx  H$. Specifically, $\Lambda(y;\theta_l) P(y;\theta_k) \approx  P(y;\theta_k)$ when $\theta_l\lesssim\theta_k$ and $0$ otherwise. The only pathological case excluded from these theorems is when the conjunctive logistic and conjunctive RBF centers exactly match in some dimension. Distance from the pathology becomes relevant when steepness parameters, $\alpha$, are small. 

Now, we approximate the product of two conjunctive RBFs, completing all the possible combinations of pairwise products between elements of our mixed basis. For this approximation (see Fig. \ref{fig:rbftsig}\textbf{b}), we have that 
\begin{equation}\label{eq:approx3}
    P(y;\theta_l)P(y;\theta_k) \approx 0.
\end{equation}

Theorem \ref{thm:augSILLconv4} shows that the approximation in Eq. (\ref{eq:approx3}) is a good approximation in the limit of $\alpha$.



\begin{theorem}\label{thm:augSILLconv4} 
When the measurements do not exactly align with the centers of the dictionary functions in any dimension, then the product of two conjunctive RBFs converges exponentially to zero as their steepness increases. 
Specifically, if $y_i\neq \mu_{ki}$ for all $i\in\{1,2,...,m\}$ and $k\in\{N_L+1,N_L+2,...,N\}$, then as $\alpha\rightarrow\infty$, $P(y;\theta_l)P(y;\theta_k)\rightarrow 0$ exponentially.
\end{theorem}

We now have four approximation theorems for mathematical terms which arise when computing the Lie derivatives of augSILL basis functions. We now apply them to these Lie derivatives with a set of approximation corollaries.

\subsection{Showing bilinear Lie derivatives can be approximated linearly to satisfy the Koopman generator equation}\label{sec:convThms3}

Corollary \ref{cor:logApprox} approximates the Lie derivative of a conjunctive logistic function in the context of our mixed basis. The approximation that this corollary suggests is not a Koopman model itself (see Fig. \ref{fig:paper_summary}). In Section \ref{sec:augSILL_error} we approximate this intermediate approximation with a full Koopman model.

\begin{corollary}\label{cor:logApprox}
Under the assumptions of Theorem \ref{thm:SILLconv}, the Lie derivative of a conjunctive logistic function exponentially approaches a nonlinear combination of augSILL functions, specifically,
\be\label{eq:augSILL_limApproxLog}
\dot{\Lambda}(y&;\theta_l) \rightarrow \sum_{i=1}^{m}\sum_{j=1}^{N_L} \alpha_{li}w_{ij}(1 - \lambda(y_i;\theta_{li})) \Lambda(y;\theta^*)  \\& +\sum_{i=1}^{m}\sum_{k=N_L+1}^{N} \alpha_{li}w_{ik}(1 - \lambda(y_i;\theta_{li})) H(y;\theta_{l},\theta_{k})
\ee exponentially as $\alpha\rightarrow\infty$.
\end{corollary}

Corollary \ref{cor:rbfApprox} approximates the Lie derivative of a conjunctive RBF in the context of the augSILL basis. 

 \begin{corollary}\label{cor:rbfApprox}
 Under the assumptions of Theorem \ref{thm:SILLconv}, the Lie derivative of a conjunctive RBF exponentially approaches a nonlinear combination of conjunctive RBFs, specifically,
 \be \label{eq:augSILL_limApproxRbf}
\dot{P}(y;\theta_l)& \rightarrow  \sum_{i=1}^{m}\sum_{j=1}^{N_L} \alpha_{li}w_{ij}(1 - 2\lambda(y_i;\theta_{li})) H(y;\theta_{j},\theta_{l})
\ee exponentially as $\alpha\rightarrow \infty$.
 \end{corollary}

Note that Equations (\ref{eq:augSILL_limApproxLog}) and (\ref{eq:augSILL_limApproxRbf}) are not compatible with the Koopman model we seek to learn: $K\psi(y),$ where the dictionary functions, $\psi$, are augSILL functions and $K$ is a real-valued matrix. In Section \ref{sec:augSILL_error} we approximate Equations (\ref{eq:augSILL_limApproxLog}) and (\ref{eq:augSILL_limApproxRbf}) with a mathematical form that is consistent with the approximated Koopman generator Eq. (\ref{eq:koop_approx_behavior}).

One can approximate the Lie derivatives of an augSILL dictionary function as a linear combination of {\em products of pairs} of augSILL functions. Below, we approximate this weighted sum of products (one of the intermediate approximations in Fig. \ref{fig:paper_summary}) with a weighted sum of augSILL functions.  The final weighted sum is of the form of Eq. (\ref{eq:koop_approx_behavior}), and therefore admits a Koopman operator model.

\begin{corollary}\label{cor:logH}
Under the assumptions of Theorem \ref{thm:SILLconv}, the sum of products
\be
&\sum_{i=1}^{m}\sum_{j=1}^{N_L} \alpha_{li}w_{ij} \Lambda(y;\theta_{l})\Lambda(y;\theta_{j})  \\& +\sum_{i=1}^{m}\sum_{k=N_L+1}^{N} \alpha_{li}w_{ik} \Lambda(y;\theta_{l})P(y;\theta_{k}) \\
\ee
approaches
\be
& \sum_{i=1}^{m}\sum_{j=1}^{N_L} \alpha_{li}w_{ij} \Lambda(y;\theta^*)  \\& +\sum_{i=1}^{m}\sum_{k=N_L+1}^{N} \alpha_{li}w_{ik} H(y;\theta_{l},\theta_{k}),
\ee a weighted sum of augSILL functions, exponentially as $\alpha\rightarrow\infty$.
\end{corollary}
\begin{proof}
This result is a direct consequence of Theorems \ref{thm:SILLconv}, \ref{thm:augSILLconv1} and \ref{thm:augSILLconv2}.
$\blacksquare$\end{proof}

\begin{corollary}\label{cor:Hrbf}
Under the assumptions of Theorem \ref{thm:SILLconv}, the sum of products
\be
&\sum_{i=1}^{m}\sum_{j=1}^{N_L} \alpha_{li}w_{ij} P(y;\theta_{l})\Lambda(y;\theta_{j})  \\& +\sum_{i=1}^{m}\sum_{k=N_L+1}^{N} \alpha_{li}w_{ik} P(y;\theta_{l})P(y;\theta_{k}) \\
\ee
approaches
\be 
& \sum_{i=1}^{m}\sum_{j=1}^{N_L} \alpha_{li}w_{ij} H(y;\theta_{j},\theta_{l}),  
\ee a weighted sum of conjunctive RBFs, exponentially as $\alpha\rightarrow\infty$.
\end{corollary}
\begin{proof}
This result is a direct consequence of Theorems \ref{thm:augSILLconv1}, \ref{thm:augSILLconv2} and \ref{thm:augSILLconv4}.
$\blacksquare$\end{proof}

The resulting linear combinations from Corollaries \ref{cor:logH} and \ref{cor:Hrbf} can be stacked and combined into the matrix $K$ (whose entries would be the products of $\alpha_{**}$ and $w_{**})$.

\begin{table*}[ht]
    \centering
    \begin{tabular}{|p{6mm}|p{8mm}|p{50mm}|p{62mm}|p{28.5mm}|}
    \hline
       \textbf{Fun.} & \textbf{Ref.}  &  \textbf{Approximation} & \textbf{Difference (Error)} & \textbf{Error Bound} \\
    \hline
       Log.  & (\ref{eq:augSILL_limApproxLog})   & \(\begin{aligned} &\sum_{i=1}^{m}\sum_{j=1}^{N_L} \alpha_{li}w_{ij} \Lambda(y;\theta^*)  \\ &+\sum_{i=1}^{m}\sum_{k=N_L+1}^{N} \alpha_{li}w_{ik} H(y;\theta_{l},\theta_{k}) \end{aligned}\) & \(\begin{aligned} &\sum_{i=1}^{m}\sum_{j=1}^{N_L} \alpha_{li}w_{ij} \lambda(y_i;\theta_{li}) \Lambda(y;\theta^*) \\& +\sum_{i=1}^{m}\sum_{k=N_L+1}^{N} \alpha_{li}w_{ik} \lambda(y_i;\theta_{li}) H(y;\theta_{l},\theta_{k}) \end{aligned}\)  & \(\begin{aligned}&\sum_{i=1}^{m}\sum_{j=1}^{N_L}\frac{\nu_{ij}}{2^{m+1}} \\&+ \sum_{i=1}^{m} \sum_{k=N_L+1}^{N} \frac{\nu_{ik}}{2^{3m+1}}\end{aligned}\) \\
    \hline
        Log. & (\ref{eq:augSILL_logPrime})  & \(\begin{aligned}  &\sum_{i=1}^{m}\sum_{j=1}^{N_L} \alpha_{li}w_{ij} \Lambda(y;\theta_l)\Lambda(y;\theta_j)  \\ &+\sum_{i=1}^{m}\sum_{k=N_L+1}^{N} \alpha_{li}w_{ik} \Lambda(y;\theta_{l})P(y;\theta_{k}) 
    \end{aligned}\)  & \(\begin{aligned} 
&\sum_{i=1}^{m}\sum_{j=1}^{N_L} \alpha_{li}w_{ij} \lambda(y_i;\theta_{li})\Lambda(y;\theta_l)\Lambda(y;\theta_j)  \\ &+\sum_{i=1}^{m}\sum_{k=N_L+1}^{N} \alpha_{li}w_{ik} \lambda(y_i;\theta_{li})\Lambda(y;\theta_{l})P(y;\theta_{k})
\end{aligned}\)  & \(\begin{aligned}&\sum_{i=1}^{m}\sum_{j=1}^{N_L}\frac{\nu_{ij}}{2^{2m+1}} \\&+ \sum_{i=1}^{m}\sum_{k=N_L+1}^{N}\frac{\nu_{ik}}{2^{3m+1}}\end{aligned}\) \\
    \hline
        RBF  & (\ref{eq:augSILL_limApproxRbf})  & \(\begin{aligned} &\sum_{i=1}^{m}\sum_{j=1}^{N_L} \alpha_{li}w_{ij} H(y;\theta_{j},\theta_{l})\end{aligned}\) & \(\begin{aligned} \sum_{i=1}^{m}&\sum_{j=1}^{N_L} 2\alpha_{li}w_{ij} \lambda(y_i;\theta_{li}) H(y;\theta_{j},\theta_{l}) \end{aligned}\) & \(\begin{aligned}\sum_{i=1}^{m}\sum_{j=1}^{N_L}\frac{\nu_{ij}}{2^{3m+1}}\end{aligned}\)  \\
    \hline
       RBF   & (\ref{eq:augSILL_rbfPrime})  & \(\begin{aligned} &\sum_{i=1}^{m}\sum_{j=1}^{N_L} \alpha_{li}w_{ij} \Lambda(y;\theta_{j})P(y;\theta_{l}) \\&+\sum_{i=1}^{m}\sum_{k=N_L+1}^{N} \alpha_{li}w_{ik} P(y;\theta_i)P(y;\theta_k)\end{aligned}\) & \(\begin{aligned}
&\sum_{i=1}^{m}\sum_{j=1}^{N_L} \alpha_{li}w_{ij} \lambda(y_i;\theta_{li})\Lambda(y;\theta_{j})P(y;\theta_{l}) \\&+\sum_{i=1}^{m}\sum_{k=N_L+1}^{N} \alpha_{li}w_{ik} \lambda(y_i;\theta_{li})P(y;\theta_i)P(y;\theta_k)
\end{aligned}\)  & \(\begin{aligned}&\sum_{i=1}^{m}\sum_{j=1}^{N_L}\frac{\nu_{ij}}{2^{3m+1}} \\&+ \sum_{i=1}^{m}\sum_{k=N_L+1}^{N}\frac{\nu_{ik}}{2^{4m+1}}\end{aligned}\) \\
    \hline
    \end{tabular}
    \caption{Approximations to and properties of error bounds for the four equations referred to in the \textbf{Ref.}. The reference equation is approximated as the corresponding equation in the \textbf{Approximation} column. We give the error of this approximation in the \textbf{Difference (Error)} column. The \textbf{Error Bound} column gives a bound on this error. The \textbf{Description} column refers to the type of dictionary function approximated in the row. The right side of Fig. \ref{fig:paper_summary} shows where these approximations fit into showing uniform finite approximate closure.}
    \label{tab:linearityErrorAugSILL}
\end{table*}

\subsection{Expectation of Approximation Error Vanishes}\label{sec:augSILL_error}

This section simultaneously addresses the approximation of two related mathematical objects.
\begin{enumerate}
    \item The nonlinear combination of AugSILL dictionaries from Corollaries \ref{cor:logApprox} and \ref{cor:rbfApprox} with linear combinations.  This is the lower left step in the right side of Fig. \ref{fig:paper_summary}.
    \item The linear combinations discussed in Corollaries \ref{cor:logH} and \ref{cor:Hrbf} with the Lie derivatives of conjunctive logistic and RBFs respectively. This is the upper right step in the right side of Fig. \ref{fig:paper_summary}.
\end{enumerate}
In total, there are four distinct approximations, one for each of the referenced corollaries above (see Table \ref{tab:linearityErrorAugSILL}). Each case our approximation is a step closer to the Koopman model. 

To understand our approximation error we compute the expected values of a single dimensional logistic and RBF. We do so with parameters and measurement values sampled from uniform distributions defined on the interval $[-a,a]$.  We choose this statistical model for how our data and parameters are sampled, because  1) the data and parameters are assumed to belong to a bounded continuum, and 2) the uniform distribution is the maximum entropy distribution for a continuous random variable on a finite interval.  Since our error terms are weighted sums of products of these functions we, under the assumption of independence, estimate the expected value of our error terms via the linearity and product rule of expectation. 

We cannot explicitly compute the probability density function (PDF) of our logistic and RBFs, so, we compute the values of these integrals numerically. Intermediate steps and details of this approximation are in Section \ref{sec:AppendixPDF} of the Appendix. In Fig. \ref{fig:expectedVals} we show their calculated expected values and variances for symmetric uniform distributions with different values of $a$. 

We find that the expected value of a logistic function will be $1/2$ (see Fig. \ref{fig:expectedVals}). Its variance, as we sample in a wider interval, tend to the functional extremes of zero and one.  This is favorable for the linearity of our approximation since, for all $\varepsilon\in (0, 0.5]$, $(0.5-\varepsilon)(0.5+\varepsilon) = 0.25 - \varepsilon^2 < 0.25 = (0.5)^2$. So, products of more extreme samples are lower in value than products of samples near the expected value. The expected value of an RBF will be no greater than $1/4$, and it decreases as $a$ increases.

\begin{figure}[ht]
    \includegraphics[width=\linewidth ]{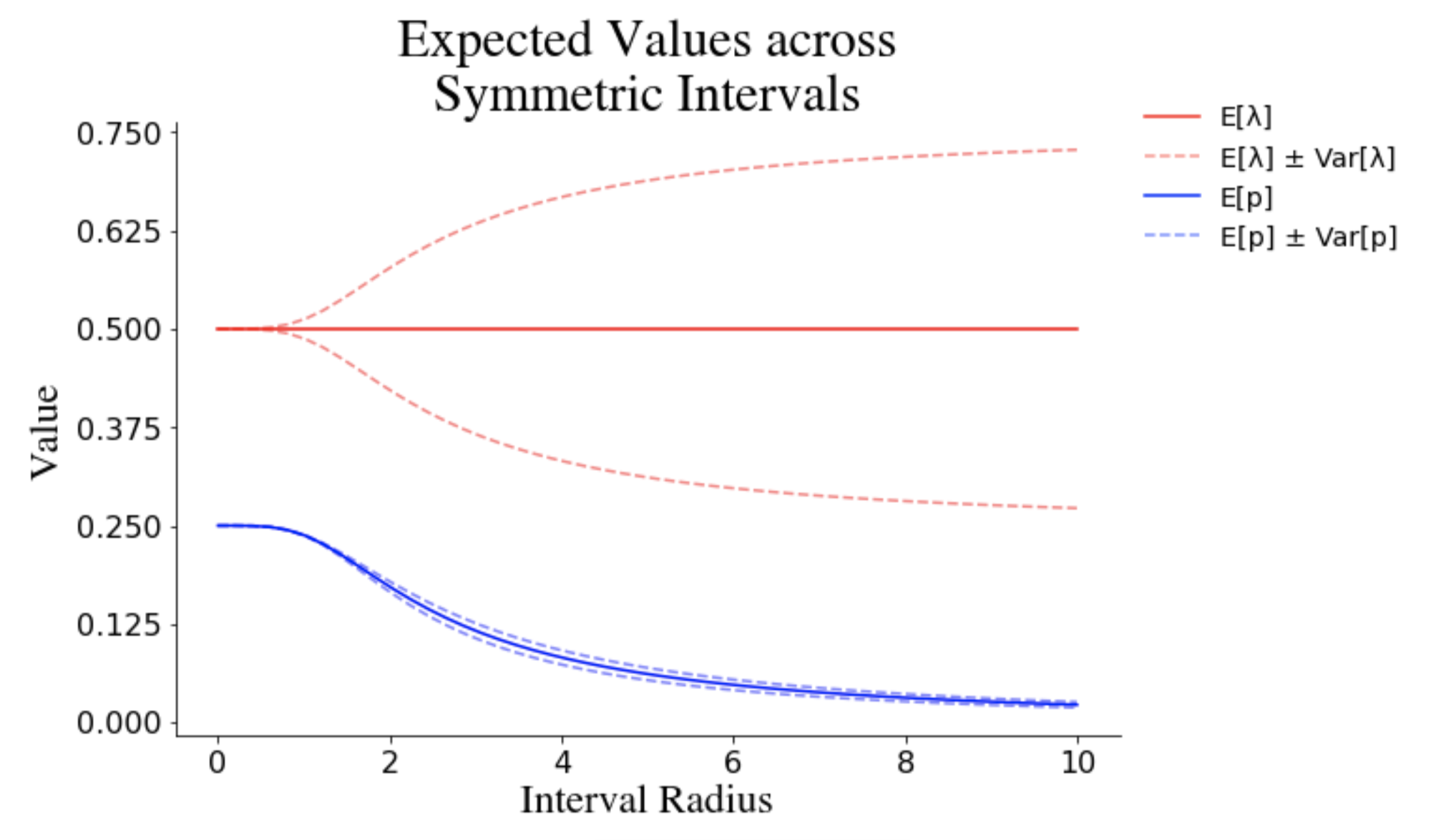}
    \caption{Expected values and variances of logistic and RBFs with parameters and measurement values sampled from symmetric uniform distributions of various interval radii. Note that the expected value of a logistic function is always $1/2$, and that of the RBF is bounded above by $1/4$.}
    \label{fig:expectedVals}
\end{figure}


We approximate a single term in the sum of the error function and extrapolate via the sum and product rule of expectation under the assumption of independence to see how nearly linear our approximation is (see Section \ref{sec:AppendixPDF} of the Appendix). We can conservatively bound the expectation of approximation error as a product that decreases exponentially with the number of measurements.  The error bounds for a conjunctive logistic and RBF are \be\label{eq:CEBlogRBF} E[\Lambda]<\frac{1}{2^m}\mbox{ and } E[P]<\frac{1}{4^m}=\frac{1}{2^{2m}}.\ee Since, $H$ will be a conjunctive RBF in $\frac{1}{2^m}^{th}$ of the measurement-parameter space, its weight of decrease can be bounded by \be\label{eq:CEBH} E[H]<\left(\frac{1}{2^{2m}}\right)\left(\frac{1}{2^m}\right)=\frac{1}{2^{3m}}.\ee
We record the full error terms and bounds in Table \ref{tab:linearityErrorAugSILL}. 


In summary, the error, $\epsilon_l(y)$, is as well or better behaved for augSILL than SILL dictionaries. So, there is a uniform bound on $\epsilon_l$, $B>0$ for augSILL dictionaries, much like there is for a standard SILL dictionary.  This means that augSILL models must satisfy uniform finite approximate closure. 

Furthermore, augSILL dictionaries, in the limit of an increasing number of measurements and increasing steepness of their dictionary functions, have that their bounding constant, $B>0$, approaches zero in expectation, $B\rightarrow 0$. Since the uniform bound on the error of this model can be arbitrarily small, the augSILL dictionaries can be used to build accurate Koopman models for any dynamic system of the from of Eq. (\ref{eq:nonlinear_system}).

\section{Numerical Examples}\label{sec:numerical}
In this article we demonstrate the uniform finite approximate closure of the SILL and augSILL dictionaries.  Models with this subspace invariance property should make accurate multi-step predictions and should generalize to data that the model has not seen before. In this section, we test the performance of our newly constructed homogeneous (SILL) and heterogeneous (augSILL) dictionaries.

To come full circle, we need to reexamine deepDMD and compare it to these models. DeepDMD uses a feedforward neural network to simultaneously parameterize the matrix $K$ as well as the dictionary functions $\psi(y)$. deepDMD has built accurate predictive models and its dimension ($N\in\mathbb{Z}^+$) scales well with that of the modeled system \cite{yeung2019learning}. We build a novel comparison of deepDMD and a much lower-parameter model built from the augSILL basis.  The lower parameter augSILL model learns as quickly and accurately as the deep-learning-based model.

To explain why algorithms like EDMD have variable success, we contribute a head to head comparison of five dictionaries for Koopman learning.  We see the deep-learning-inspired dictionaries vastly outperform orthogonal polynomial dictionaries. This suggests that issues with algorithms like EDMD may be resolved by selecting a dictionary proven to satisfy uniform finite approximate closure.

Unless specified otherwise we use simulated data generated from uniformly distributed initial states run with SciPy's ODE integration software. Each of these results is concerned with the discrete-time, data-driven problem statement. The system's state is directly measured at even time intervals.

\begin{figure*}[ht]
    \centering
    \includegraphics[width=0.7\linewidth ]{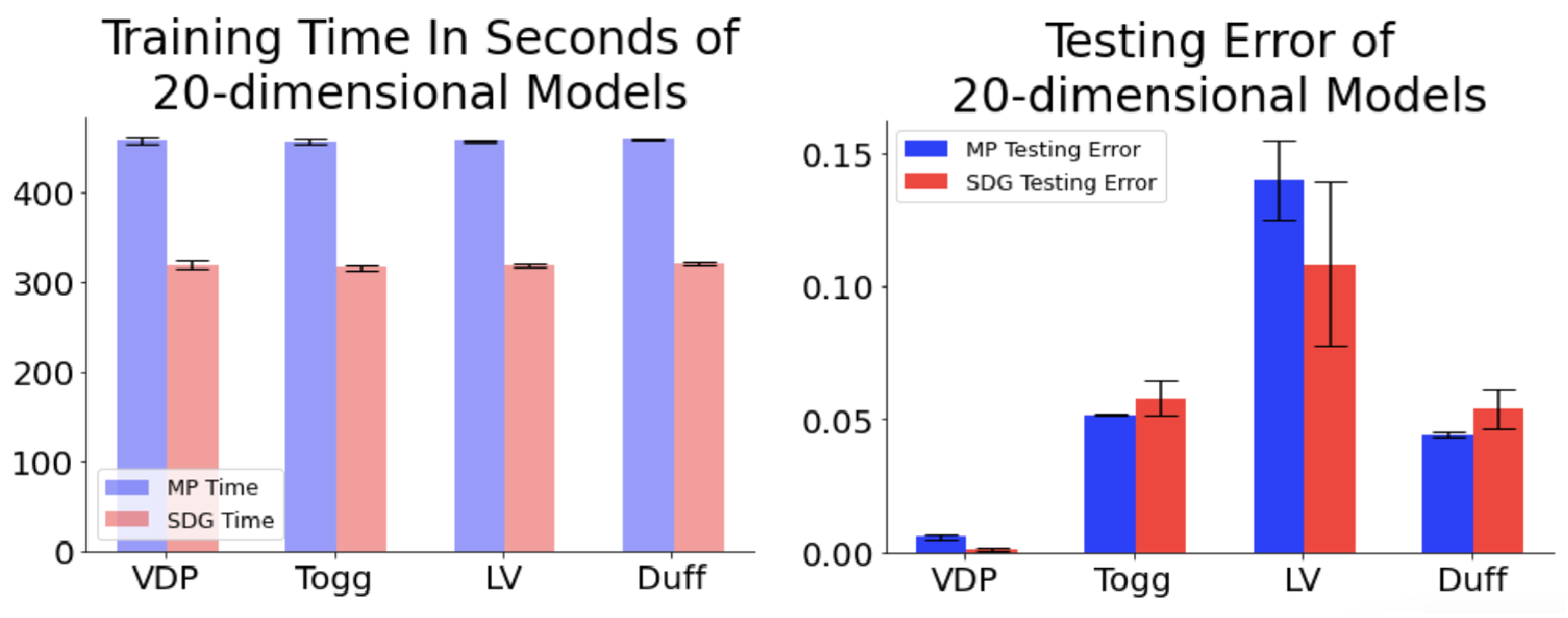}
    \caption{A comparison of the full matching pursuit algorithm to SGD on four nonlinear systems. We plot 5-step prediction error. The SGD algorithm (in red) shows consistently better temporal scaling as well as comparable 5-step prediction error.}
    \label{fig:mp_vs_sgd}
\end{figure*}

\subsection{Choosing center parameters}
Even when a dictionary class and model dimension are selected, each individual problem will warrant a unique parameterization of the dictionary. Given a dictionary, such as the augSILL dictionary, how do we choose the parameters of each dictionary function? We consider two algorithms, matching pursuit and stochastic gradient descent (SGD).


Matching pursuit \cite{mallat1993matching} considers an expansive list of potential dictionary functions and greedily adds the function that lowers the value of the objective function (Eq. (\ref{eq:objective})) most. Matching pursuit adds one function at a time to the model.

We use SGD to attack a host of non-convex optimization problems. It is famous, in part, because of its use in training artificial neural networks. SGD can be directly applied to parameterize a fixed number of dictionary functions from data, much like it learns the parameters of a neural network.

\subsubsection{Which algorithm do we choose?}\label{sec:numerical-comp}
We compared two variations of matching pursuit, as well as SGD for learning augSILL models of four dynamic systems, the Van der Pol oscillator, the Duffing oscillator, the Lokta-Volterra model and the Gardner-Collins toggle switch. The specific parametrization of these systems is given in Section \ref{sec:AppendixSystems} of the Appendix.  The measurements for each system are the state variables themselves.

We found that the full matching pursuit algorithm and SGD had similar performance for a 20 dimensional Koopman operator (see Fig. \ref{fig:mp_vs_sgd}).  
We focus on accuracy and performance for 20 dimensional models as a step-in for modeling higher dimensional systems.  Also, we note that SDG was about 1.5 times faster when building a 20 dimensional model. Since SGD was the better choice for building larger Koopman models in terms of time to execute, and performed comparably in 5-step prediction error, we compare this algorithm to deepDMD. Since deepDMD utilizes SGD, we can compare model accuracy at each training epoch.

\begin{figure}[ht]
    \centering
    \includegraphics[width=\linewidth ]{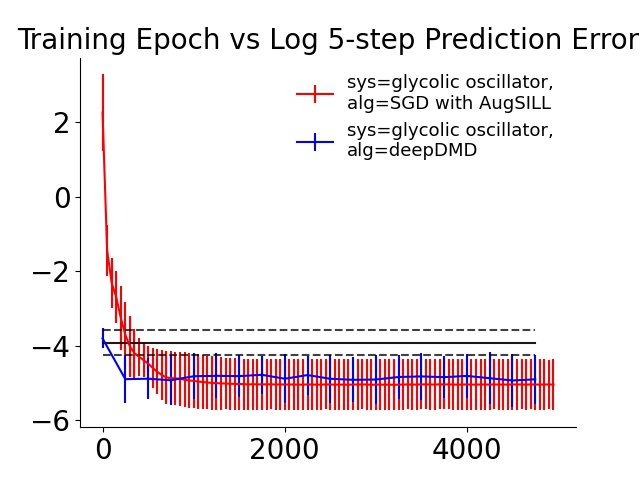}
    \caption{A comparison deepDMD to SGD with the augSILL basis on the seven dimensional model of glycolysis given in \cite{daniels2015efficient}. The solid horizontal bar is the mean performance of the DMD algorithm, the dotted bars are one standard deviation above and below this mean. The spines are error bars for the measured epochs.}
    \label{fig:DeepDMDvsSGD}
\end{figure}

\subsection{AugSILL basis vs deepDMD}
The SILL and augSILL dictionaries are a targeted study of the dictionary functions generated from deepDMD. Our visualization of deepDMD observables showed a convergence to sums of SILL and augSILL dictionary observables. Can we use these dictionaries to build a model on par with deep learning?

To challenge ourselves, we used a seven dimensional glycolysis model to generate our testing and training data \cite{daniels2015efficient}. This data was all generated from a single initial condition $x_0=[1,0.19,0.2,0.1,0.3,0.14,0.05]$.  The augSILL model reaches a comparable 5-step prediction error to deepDMD in under 1000 training epochs and neither significantly changes over the next 4000 epochs, see Fig. \ref{fig:DeepDMDvsSGD}. Note that the model we learn using the augSILL basis has $995$ parameters.  All in all, the deepDMD model has a total of $3,949$ parameters. All of the augSILL parameters are easily interpreted as center, steepness and weight parameters of a logistic or RBF.

\begin{figure*}[ht]
    \centering
    \includegraphics[width=.85\linewidth ]{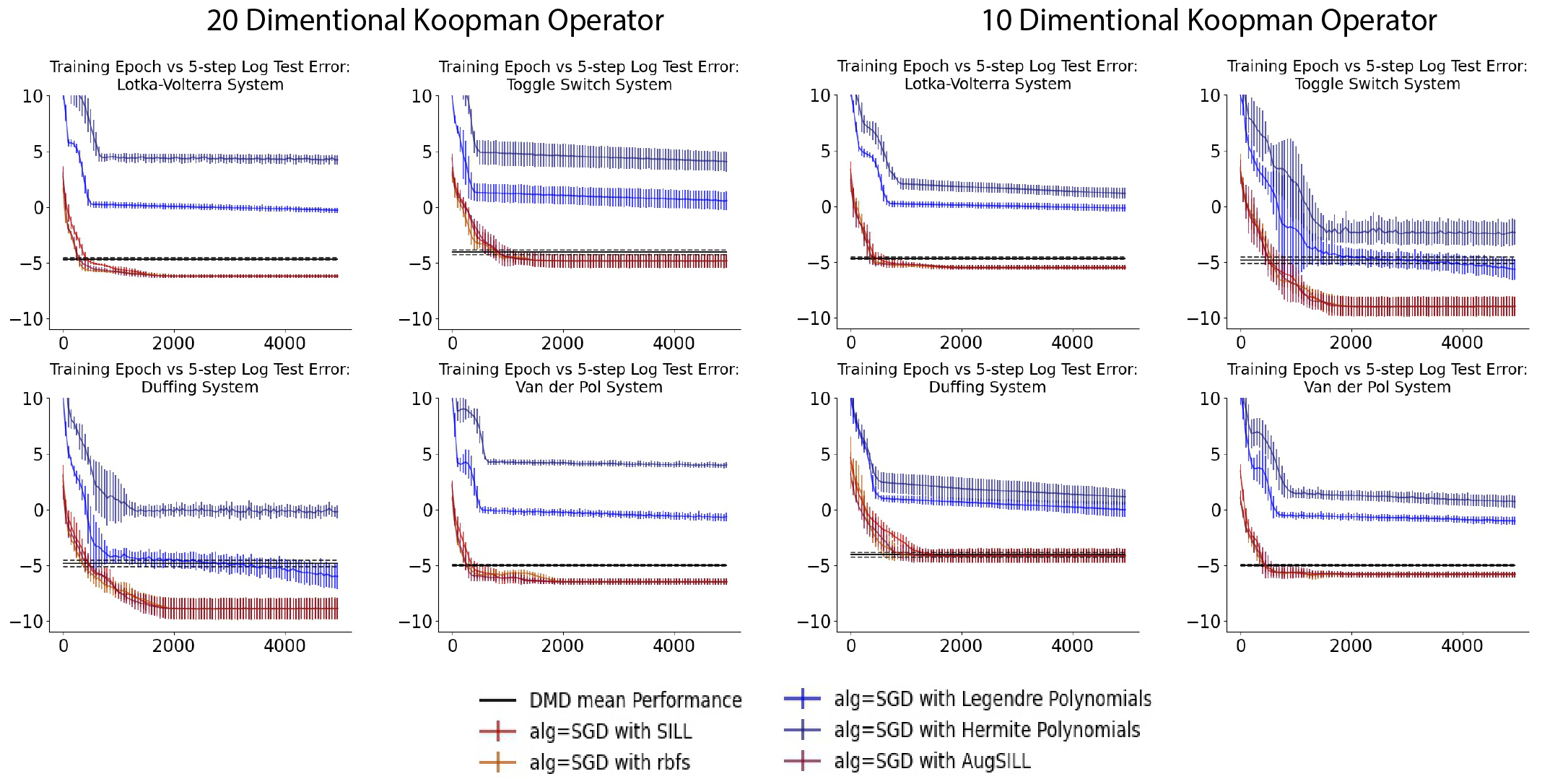}
    \caption{Five-step prediction accuracy vs training epoch of Koopman learning with SGD using various dictionaries. Plotted on a log scale. Note that the Hermite and Legendre polynomial basis (in blue) have greater error throughout the training process. The spines are error bars at each epoch where 5-step error was measured.} 
    \label{fig:algComp}
\end{figure*}

\subsection{Comparison of dictionaries for Koopman learning}
Now we address the relationship between the choice of dictionary and the success of a Koopman model.  What properties do successful Koopman dictionaries have in common?  The augSILL basis compared in performance to deepDMD.  Would we have gotten similar results using other dictionaries?

We learn the systems in Eq. (\ref{eq:vanDerPol}), (\ref{eq:duff}), (\ref{eq:pred}), and (\ref{eq:togg}), parameterized with the SGD algorithm. We do so with the SILL, augSILL and summed one-dimensional RBFs, as well as two different orthogonal polynomial dictionaries (Legendre and Hermite). We compare their 5-step average prediction error for a 5, 10 and 20 dimensional Koopman model of each system. 

For the 5 dimensional models the choice of dictionary seemed mostly irrelevant. However, for a 5 dimensional model, SGD only outperformed standard DMD for the toggle switch. So, we don't have enough system dimensions with these dictionaries to want to use SGD in the first place. 

Building a 10 and 20 dimensional Koopman model of each system, we find that the basis inspired from the outputs of deepDMD (SILL, augSILL and summed RBFs) have lower 5-step prediction errors (see Fig. \ref{fig:algComp}).  Each of these three basis had nearly identical errors, though the SILL basis can take more training iterations to perform comparably. The Legendre polynomial basis outperformed the Hermite polynomial basis in each case.  For the 20 dimensional model of the Van der Pol system the augSILL basis was over 372 times more accurate than the Legendre polynomial basis, for the Duffing Oscillator it was over 221 times more accurate, for the Predator-Prey System it was over 18 times as accurate and for the Toggle Switch it was over 328 times as accurate.  


%
%
\section{Conclusion} 
Learning models in a data-driven setting using human-defined dictionaries results in high-dimensional, over-parameterized representations of what could be simple physical phenomena. DeepDMD and other ANN learning techniques address these issues but at the cost of extensive computational time. Moreover, the dictionaries learned by deepDMD are adhoc and randomly constructed, heterogeneous in nature, and difficult to interpret.  To date, no systematic method has explored using mixed or heterogeneous dictionaries to improve performance in Koopman learning problems. 

To investigate how deepDMD finds successful Koopman invariant subspaces, we extracted simple features of the dictionary functions learned by deepDMD and looked for properties to explain their success.  We discovered a mixed Koopman dictionary which we call the augSILL dictionary. We showed, in Section \ref{sec:AugSILLclosure}, that augSILL model error drops exponentially when steepness parameters increase and more measurements are included in the data.  Quantifying how much should be measured to build an accurate Koopman model is an intriguing future research direction.

The success of any measurement-inclusive Koopman dictionary depends on more than how the dictionary approximates the evolution of the measurements in time. It also depends on how well linear combinations of the dictionary functions approximate the dictionary function's Lie derivatives. In under 1000 training epochs augSILL models matched the accuracy of trained deepDMD models. These models were anywhere from 18 to 372 times as accurate as models made using the polynomial dictionaries, which have no guarantees of uniform finite approximate closure. The augSILL dictionary performs like deepDMD to learn a dynamic system from data using an order of magnitude fewer parameters.  Further, the augSILL dictionary is fully specified with closed-form analytical expressions (unlike deepDMD) and as a consequence of the theoretical results in this paper, satisfies a unique numerical property of uniform approximate finite closure.  Our methodology provides a template for understanding how deep neural networks successfully approximate governing equations \cite{brunton2016discovering}, the action of operators and their spectra \cite{mezic2005spectral}, and dynamical systems \cite{mezic2013analysis}.  Further, these results provide a pattern for improving scalability and interpretability of dictionary-based learning models for dynamical system identification. 

\begin{ack}
Any opinions, findings and conclusions or recommendations expressed in this material are
those of the author(s) and do not necessarily reflect the
views of the Defense Advanced Research Projects Agency
(DARPA), the Department of Defense, or the United States
Government. This work was supported partially by a Defense Advanced Research Projects Agency (DARPA) Grant
No. FA8750-19-2-0502, PNNL Grant No. 528678, ICB Grant No. W911NF-19-D-0001 and No. W911NF-19-F-0037, and ROE Young Investigator Grant No. W911NF-20-1-0165.
\end{ack}

\bibliographystyle{plain}        
\bibliography{bibliography}           

\appendix
\section{Notation}\label{sec:AppendixNotation}
Our models, throughout this article, will have two classes of parameters. Each class refers to distinct geometric properties. To create a clearer separation of function variables and parameters we will use the following notation \[\Eta(x;y_l,z_l).\]  In this notation, $\Eta$ is a function whose variable is the vector $x$ and whose parameters are vectors $y_l$ and $z_l$. The vectors $y_l$ and $z_l$ refer to the geometrically distinct classes of parameters.  A second example would be \[\eta(x_i;y_{li}, z_{li}).\] In this notation, $\eta$ is a function whose variable is the scalar $x_i$, and whose parameters are the scalars $y_{li}$ and $z_{li}$.  To make our equations less cumbersome we summarize all the distinct parameters with the label $\theta$. For example, \[\Eta(x;y_l,z_l)\triangleq \Eta(x;\theta_l),\] and \[\eta(x_i;y_{li}, z_{li})\triangleq \eta(x_i;\theta_{li}).\]

In general, our notation uses the following conventions. \begin{enumerate}
    \item Integer $i$ will be an index of measurement dimension.
    \item Integer $j$ will be an index of the first group of added dimensions.
    \item Integer $k$ will be an index of the second group of added dimensions.
    \item Integer $l$ will be an arbitrary added dimension index.
    \item Integer $n$ will be the state dimension.
    \item Integer $m$ will be the number of  measurements.
    \item Integer $N$ will be the added dimensions.  When we mix two basis, the dimension of the first (conjunctive logstic functions) will be $N_L$ and the dimension of the second (conjunctive radial basis functions) will be $N_R$.  This means that $N_L + N_R = N$.
    \item The $m\times N$ real-valued matrix, $w$, will be a matrix of weights. In our analysis it corresponds to a block of the Koopman Operator approximation matrix, $K$, which describes the flow of the state variables as a linear combination of nonlinear observables.
\end{enumerate}

\section{Average SILL Error Induced by Linearity of the Model and Final Results on Subspace Invariance of SILL}\label{sec:averageErrorSILL}

To show uniform finite approximate closure, we need to characterize the error, \be\label{eq:overallError} \epsilon(y) = \frac{d\psi(y)}{dt} - K\psi(y).\ee  We do so by choosing a specific approximation for the time derivative of each nonlinear dictionary function in the SILL basis.

Characterizing the closure of a dictionary means understanding the error between the Koopman model built with that dictionary and the true Lie derivative of each dictionary function. In Section \ref{sec:SILLclosure} we used mathematical analysis to show that the approximation error of some models will go to zero in the limit of high steepness. These models do not fully bridge the gap between our Koopman model and the true Lie derivative. This section characterizes additional approximation errors to fully connect our Koopman model to the true Lie derivative. Then, we uniformly bound those errors to show the uniform finite approximate closure of the SILL dictionary.

The error of approximating Eq. (\ref{eq:SILL_nonlinComb}) with Eq. (\ref{eq:SILL_finalApprox}) is: 
\be\label{eq:SILL_lin_error}
 \sum_{i=1}^{m}\sum_{j=1}^{N}\alpha_{li}w_{ij}\lambda(y_i;\theta_{li})\Lambda(y;\theta^*), 
\ee  and the error of approximating Eq. (\ref{eq:SILL_lie_derivative}) with Eq. (\ref{eq:SILL_linearLie}) is: \be\label{eq:SILL_linProd_error}
 \sum_{i=1}^{m}\sum_{j=1}^{N}\alpha_{li}w_{ij}\lambda(y_i;\theta_{li})\Lambda(y;\theta_l)\Lambda(y;\theta_j).
\ee

To get a grasp on the kind of errors we can expect when modeling with the SILL basis we need to get a grasp on the sizes of Eq. (\ref{eq:SILL_lin_error}) and Eq. (\ref{eq:SILL_linProd_error}).

We do so by sampling possible parameter and state values from uniform distributions over a symmetric interval and then computing distributions of one dimensional logistic functions.  We then do the same for terms in the sums in Equations (\ref{eq:SILL_lin_error}) and  (\ref{eq:SILL_linProd_error}), under the assumption that $w_{ij}=1$. In doing so, we notice an exponential decrease in expected error as the number of measurements increases. Details are given in Section \ref{sec:augSILL_error}.  But, the end result is that one can conservatively expect the error of each summation term in Eq. (\ref{eq:SILL_lin_error}) to decrease in the numbers of measurements, $m$, at a rate bound above by $\frac{1}{2^{m+1}}$ and  the error of each summation term in Eq. (\ref{eq:SILL_linProd_error}) to decrease in $m$ at a rate bound above by $\frac{1}{2^{2m+1}}$.

Explicitly, these error bounds are \be\label{eq:SILLerrorBounds} \sum_{i=1}^{m}\sum_{j=1}^{N}\frac{\nu_{ij}}{2^{m+1}} \mbox{ and }\sum_{i=1}^{m}\sum_{j=1}^{N}\frac{\nu_{ij}}{2^{2m+1}}, \ee where $\nu_{ij}\in\R$.

\subsection{Error of the final model}\label{sec:VC}
Theorem \ref{thm:SILLconv} shows that in the limit of infinitely steep dictionary functions, the error of approximating the product of two conjunctive logistic functions with a single one goes to zero exponentially.  Likewise, our probabilistic bound demonstrated that the error in approximating the Lie derivative of SILL functions as a Koopman model goes to zero exponentially, on average, as the number of measurements increases.  This means that the error, $\epsilon_l(y)$, of the approximation 
\be \label{eq:SILLapproximation}
\dot{\Lambda}(y;\theta_l)& =
\sum_{i=1}^{m}\sum_{j=1}^{N}\alpha_{li}w_{ij}(1-\lambda(y_i;\theta_{li}))\Lambda(y;\theta_l)\Lambda(y;\theta_j)\\& \approx \sum_{i=1}^{m}\sum_{j=1}^{N}\alpha_{li}w_{ij}\Lambda(y;\theta^*)
\ee exponentially go to zero (and therefore will be uniformly bounded by a constant $B>0$) with increasingly steep logistic functions and an increasingly large number of measurements.  As an aside, these error bounds can be used as a guiding upper bound for the number of observables to measure to build a good Koopman model.

So the bound, $B$, on our closure error, $\epsilon(y)$, goes to zero, $B\rightarrow 0$ at the extremes of large measurements and steep dictionary functions. However, in practice, one captures strong nonlinear system features using the SILL basis for a low dimensional dictionary and a few measurements \cite{johnson2018class}.

\section{Proofs of Theorems}\label{sec:AppendixProofs}

In \cite{johnson2018class} we demonstrated that imposing a total order on conjunctive logistic basis functions implied some basic closure properties of the SILL dictionary.  We show in this section proofs that apply this total order as steps to showing the uniform finite closure of the SILL and augSILL dictionaries. The needed total order is summed up in the following assumption.
\begin{assumption}\label{assump:order}
There exists a total order on the set of conjunctive logistic functions, $\Lambda(y;\theta_1) ,..., \Lambda(y;\theta_l)$, induced by the positive orthant $\mathbb{R}^m_+$, where $\theta_{l} \gtrsim \theta_{j}$ whenever $\mu_j - \mu_l \in \mathbb{R}^m_+$, and therefore $\Lambda(y;\theta_l) \geq \Lambda(y;\theta_j)$. 
\end{assumption}  Without loss of generality, we choose to label our parameters to impose that $\theta_l \lesssim \theta_j$ whenever $l \leq j$.

In these results we do not consider the error of approximating the vector field, $F$, we instead focus on the subspace invariance of a model built from a dictionary that already captures the basic system dynamics. This assumption is given formally below. 

\begin{assumption}
The vector field, $F$, lies in the span of our SILL (or alternativly, augSILL) dictionary.
\end{assumption}

\textit{Theorem} \ref{thm:SILLconv}:
Under Assumption \ref{assump:order}, if the dictionary functions do not exactly match their corresponding center parameters, $y_i\neq\mu_{ji} $ for all $ i\in \{1, 2, ..., m\}$ and  $j\in\{1,2,...,N\}$, then, as the steepness parameters go to infinity, the product of two conjunctive logistic function will exponentially approach a single conjunctive logistic function in the dictionary, $\alpha\rightarrow\infty$, \[\Lambda(x;\theta_l)  \Lambda(x;\theta_j) - \Lambda(x;\theta^*)\rightarrow 0\] exponentially.

\begin{proof}
We now investigate the error of the approximation. Eq. (\ref{eq:approx1})'s error of approximation is 
\begin{equation}\label{Lambda_errorterm}
\begin{aligned}
\Lambda(y&;\theta_l)  \Lambda(y;\theta_j) - \Lambda(y;\theta_l) \\ 
&=\Lambda(y;\theta_l) \left(\Lambda(y;\theta_j) -1\right)\\ 
&=\frac{1-\Lambda(y;\theta_j)^{-1}}{(\Lambda(y;\theta_l)\Lambda(y;\theta_j))^{-1}} \\
&= \frac{1 -  (1+e^{-\alpha_{l1}(y_1 - \mu_{l1})})   ... (1+e^{-\alpha_{lm}(y_m - \mu_{lm})}) }{\prod_{i=1}^{m}(1+e^{-\alpha_{li}(y_i - \mu_{li})}) (1+e^{-\alpha_{ji}(y_i - \mu_{ji})})},
\end{aligned}
\end{equation} whenever $\theta_l\gtrsim\theta_j.$

Without loss of generality, assume that $\theta_l\gtrsim\theta_j$. As we hold $y$ constant and let $\alpha\rightarrow \infty$,  for all $ i\in \{1, 2, ..., m\}$, and any possible value of $l,j$ we observe two cases. In these cases we denote a possible value of $\mu_{li}$, $\mu_{ji}$ etc. as $\mu_*$

\textbf{Case 1}, $y_i-\mu_* > 0$: As $\alpha \rightarrow \infty$ we have that $e^{-\alpha(y_i-\mu_*)}\rightarrow 0$ and so $\frac{1}{1+ e^{-\alpha_i(y_i-\mu_*)}} \rightarrow 1$. 

\textbf{Case 2}, $y_i-\mu_* < 0$: As $\alpha \rightarrow \infty$ we have that $e^{-\alpha_i(y_i-\mu_*)}\rightarrow \infty$ and so $\frac{1}{1+ e^{-\alpha_i(y_i-\mu_*)}} \rightarrow 0$ exponentially.  

So, if there exists $i\in \{1,2,...,m\}$ for every $ j$,  so that $y_i-\mu_{ji} < 0$, then Eq. (\ref{eq:approx1}) goes to 0 exponentially as $\alpha \rightarrow \infty$.  

However, if $y_i-\mu_{ji} > 0$ for all $i$, then, since $\theta_l \gtrsim \theta_j$, for all $i$ we have that $y_i-\mu_{li} > 0$ for all $i$ as well, so Eq. (\ref{Lambda_errorterm}) goes to $\frac{1-1}{1}=0$ exponentially as $\alpha \rightarrow \infty$.  
$\blacksquare$\end{proof}

\textit{Corollary}
\ref{cor:logSILLApprox}:
Under the assumptions of Theorem \ref{thm:SILLconv}, when $F$ is spanned by a SILL dictionary, the Lie derivative of a conjunctive logistic function exponentially approaches a finite weighted sum of conjunctive logistic functions as the steepnesses of the functions goes to infinity. Specifically, \be\dot\Lambda(x;\theta_l)\rightarrow \sum_{i=1}^{n}\sum_{j=1}^{N}\alpha_{li}w_{ij}(1-\lambda(x_i;\theta_{li}))\Lambda(x;\theta^*)\ee exponentially as $\alpha\rightarrow\infty$.

\begin{proof} 
We assume that $F$ is spanned by our set of nonlinear dictionary functions.

This means that there exists a real-valued weighting matrix, $w\in \R^{m\times N}$, so that for any, $i\in\{1,2,...,m\}$, the $i^{th}$ element of $F$, $F_i$, can be written as:
\begin{equation}\label{eq:f_regression}
F_i(y) = \sum_{j=1}^{N} w_{ij}  \Lambda(y, \theta_j).
\end{equation}

Therefore, the time derivative of an arbitrary nonlinear observable in the SILL dictionary is: 

\begin{equation}\label{eq:SILL_lie_derivative} 
    \begin{aligned}
    \dot \Lambda(y;&\theta_l) = (\nabla_y\Lambda(y;\theta_l))^T\frac{dy}{dt} = (\nabla_y\Lambda(y;\theta_l))^TF(y)\\
    &=\sum_{i=1}^m\alpha_{li}(1-\lambda(y_i;\theta_{li}))\Lambda(y;\theta_l)F_i(y)\\
    &=\sum_{i=1}^m\alpha_{li}(1-\lambda(y_i;\theta_{li}))\Lambda(y;\theta_l)\sum_{j=1}^{N} w_{ij}  \Lambda(y; \theta_j)\\
    &=\sum_{i=1}^{m}\sum_{j=1}^{N}\alpha_{li}w_{ij}(1-\lambda(y_i;\theta_{li}))\Lambda(y;\theta_l)\Lambda(y;\theta_j).
    \end{aligned}
\end{equation}

Theorem \ref{thm:SILLconv} shows, that in the limit of steepness, the error between Eq. (\ref{eq:SILL_lie_derivative}) and \be\label{eq:SILL_nonlinComb}\sum_{i=1}^{m}\sum_{j=1}^{N}\alpha_{li}w_{ij}(1-\lambda(y_i;\theta_{li}))\Lambda(y;\theta^*) \ee will go to zero exponentially. $\blacksquare$\end{proof}

 \textit{Theorem} \ref{thm:augSILLconv1}:
If $y_i\neq\mu_{ki} $ for all $ i\in \{1, 2, ..., m\}$ and $k\in\{N_L+1,N_L+2,...,N\}$  and there exists $i^*\in\{1, 2, ..., m\}$ so that $\mu_{ki^*} \geq \mu_{li^*}$, then as $\alpha\rightarrow\infty$, $\Lambda(y;\theta_l) P(y;\theta_k) -  H(y;\theta_l,\theta_k)\rightarrow 0$ exponentially.

\begin{proof}
As a preliminary, we note that $-1\leq (\Lambda(y;\theta_l) - 1)\leq 0$ for any $y\in\R^m$ and for $\alpha>0$.  

We first show that when $\mu_{ki} \geq \mu_{li}$ $P(y;\theta_k)\rightarrow 0$ as $\alpha\rightarrow \infty$. We do so by considering the two possible cases for any $i$ where $\mu_{ki} \geq \mu_{li}$. 

\textbf{Case 1}, $y_i>\mu_{ki}$.  In this case as $\alpha\rightarrow\infty$ we have that $e^{-\alpha_{ki}(y_i-\mu_{ki})}\rightarrow 0$ exponentially, and that $(1+e^{-\alpha_{ki}(y_i-\mu_{ki})})^2 \rightarrow 1^2=1$. Thus the $i^{th}$ term of $P(y;\theta_k)$ will go to  $\frac{0}{1}=0$ exponentially.

\textbf{Case 2}, $y_i<\mu_{ki}$.  In this case as $\alpha\rightarrow\infty$  we have that $\frac{e^{-\alpha_{ki}(y_i-\mu_{ki})}}{(1+e^{-\alpha_{ki}(y_i-\mu_{ki})})^2}\rightarrow 0$. Thus the $i^{th}$ term of $P(y;\theta_k)$ will go to $0$ exponentially.

Each of these cases implies that $P(y;\theta_k)$ and therefore Eq. (\ref{eq:P_errorterm1}) goes to $0$ exponentially.
$\blacksquare$\end{proof}

\textit{Theorem} \ref{thm:augSILLconv2}:
If $y_i\neq\mu_{ki}$ and $y_i\neq\mu_{li} $ for all $ i\in \{1, 2, ..., m\}$ and for all $k\in\{N_L+1, N_L+2, ..., N\}$ so that $\mu_{ki} < \mu_{li}$, then as $\alpha\rightarrow\infty$, $\Lambda(y;\theta_l) P(y;\theta_k) -  H(y;\theta_l, \theta_k)\rightarrow 0$ exponentially.

\begin{proof}
We assume that $\mu_{ki} < \mu_{li}$ and note that, in each case as $\alpha\rightarrow \infty$ the $i^{th}$ term in Eq. (\ref{eq:P_errorterm2}) goes to zero as shown in Case 2 of the proof of Theorem \ref{thm:augSILLconv1}.  Thus the product of these terms will go to zero as $\alpha\rightarrow\infty$.  For the sake of brevity we forgo the explicit computation of the limits which follow and speak in more general terms of rate of growth. We note that when we use ``$\infty$'' we mean to say that the term grows to infinity with $\alpha$ with a rate of $e^{c\alpha}$, for some constant $c$.  Furthermore, we use ``$0_*$'' to mean that the term goes to zero with a rate of $e^{-c\alpha}$, for some constant $c$.

\textbf{Case 1}, $y_i > \mu_{ki} $ \textit{Sub-Case 1.1}, $y_i > \mu_{li}$ so as $\alpha\rightarrow\infty$ our $i^{th}$ term goes to $\frac{0_*}{1}=0$. 
\textit{Sub-Case 1.2}, $y_i = \mu_{li}$ so as $\alpha\rightarrow\infty$ our $i^{th}$ term goes to $\frac{0_*}{2}=0$. 
\textit{Sub-Case 1.3},  $y_i < \mu_{li}$ so as $\alpha\rightarrow\infty$ our $i^{th}$ term goes to $\frac{0_*}{\infty}=0$.

\textbf{Case 2}, $y_i < \mu_{ki}$, thus $x_i<\mu_{li}$ and so as $\alpha\rightarrow\infty$ our $i^{th}$ term goes to  $\frac{\infty}{\infty^2\infty}=0$.

\textit{Sub-Case 2.1}, $y_j > \mu_{li}$ so as $\alpha\rightarrow\infty$ our $i^{th}$ term goes to $\frac{\infty}{\infty^2}=0$. 
\textit{Sub-Case 2.2}, $y_i = \mu_{li}$ so as $\alpha\rightarrow\infty$ our $i^{th}$ term goes to $\frac{\infty}{2\infty^2}=0$. \textit{Sub-Case 2.3},  $y_i < \mu_{li}$ so as $\alpha\rightarrow\infty$ our $i^{th}$ term goes to $\frac{\infty}{\infty^2\infty}=0$.

Now we note that, in each case as $\alpha\rightarrow \infty$ the $i^{th}$ term in Eq. (\ref{eq:P_errorterm2}) goes to zero.  Thus the product of these terms will go to zero as $\alpha\rightarrow\infty$.
$\blacksquare$\end{proof}

\textit{Lemma} 1:
If $||\mu_l-\mu_k||_\infty\neq 0$ or $\mu_{li}\neq \mu_{ki}$ for all $i\in\{1,2,...,m\}$ and all $k\in\{N_L+1,N_L+2,...,N\}$, then as $\alpha\rightarrow\infty$, $P(y;\theta_l)P(y;\theta_k)\rightarrow 0$ exponentially.

\begin{proof}
As we consider the approximation error of a product of two conjunctive RBFs with zero our error is the term: \begin{equation}\label{eq:PP_errorterm}
\begin{aligned}
P(y;&\theta_l)P(y;\theta_k) = \\&\prod_{i=1}^m\frac{ e^{-\alpha_{li}(y_i-\mu_{li})} e^{-\alpha_{ki}(y_i-\mu_{ki})}}{ (1+e^{-\alpha_{li}(y_i-\mu_{li})})^2(1+e^{-\alpha_{ki}(y_i-\mu_{ki})})^2 }.
\end{aligned}
\end{equation}

We proceed by looking at the $i^{th}$ term in  Eq. (\ref{eq:PP_errorterm}), we know how it develops as $\alpha\rightarrow\infty$ by examining the cases tabulated below.  As in the proof of Theorem \ref{thm:augSILLconv2} we note that when we use ``$\infty$'' we mean to say that the term grows to infinity with $\alpha$ with a rate of $e^{c\alpha}$, for some constant $c$. Furthermore, we use ``$0_*$'' to mean that the term goes to zero with a rate of $e^{-c\alpha}$ for some constant $c$.
\begin{center}
  \begin{tabular}{| l | c | c | c |}
  \hline
    \textbf{Cases} &$y_i<\mu_{li}$ & $y_i=\mu_{li}$ & $y_i>\mu_{li}$ \\ \hline
    $x_i<\mu_{ki}$ & $\frac{\infty\infty}{\infty^2\infty^2}\rightarrow 0$ & $\frac{\infty}{\infty^2}\rightarrow 0$ & $\beta_1$ \\ \hline
    $y_i=\mu_{ki}$ & $\frac{\infty}{\infty^2}\rightarrow 0$ & $\frac{1}{16}$ & $\frac{0_*}{2}\rightarrow 0$ \\
    \hline
    $y_i>\mu_{ki}$ & $\beta_2$&$\frac{0_*}{2}\rightarrow 0$&$\frac{0_*0_*}{1}\rightarrow 0$\\ \hline
  \end{tabular}
\end{center}

We first comment that the case in the center will occur in at most $m-1$ of the terms of Eq. (\ref{eq:PP_errorterm}) by our assumption that $||\mu_l-\mu_k||_\infty\neq 0$.  We then note that given $\beta_1$ we have three cases: 

\textbf{Case 1}, $-\alpha_{li}(y_i-\mu_{li})-\alpha_{ki}(y_i-\mu_{ki}) < 0$.  In this case we have that as $\alpha\rightarrow\infty$ that our term goes to $\frac{0_*}{\infty^2} \rightarrow 0.$

\textbf{Case 2}, $-\alpha_{li}(y_i-\mu_{li})-\alpha_{ki}(y_i-\mu_{ki}) = 0$. In this case we have that as $\alpha\rightarrow\infty$ that our term goes to $\frac{1}{\infty^2} \rightarrow 0.$

\textbf{Case 3}, $-\alpha_{li}(y_i-\mu_{li})-\alpha_{ki}(y_i-\mu_{ki}) > 0$. In this case we have that as $\alpha\rightarrow\infty$ that our term goes to $\frac{\infty}{\infty^2} \rightarrow 0.$

We now note that without loss of generality that these three cases cover $\beta_2$ (one simply swaps $\alpha$ terms) and so we have that as $\alpha\rightarrow \infty$ at least one term in our product will go to zero exponentially and the rest at most will go to $\frac{1}{16}$.  Thus the error term goes to zero exponentially.
$\blacksquare$\end{proof}

\textit{Theorem} \ref{thm:augSILLconv4}:
If $y_i\neq \mu_{ki}$ for all $i\in\{1,2,...,m\}$ and $k\in\{N_L+1,N_L+2,...,N\}$, then as $\alpha\rightarrow\infty$, $P(y;\theta_l)P(y;\theta_k)\rightarrow 0$ exponentially.

\begin{proof}
This proof is very similar to the proof of Lemma 1. The only difference is that from the table of cases, only the four corner cases can occur.
$\blacksquare$\end{proof}

Assume that our augSILL observables span the vector field we seek to model. Define $F_i(y)$ to be the $i^{th}$ entry of the vector field, $F$. Then $F_i(y)$ may be written as the following weighted sum (for nonnegative integers, $N_L$ and $N_R$ so that $N_L+N_R=N$):

\[F_i(y) = \sum_{j=1}^{N_L}w_{ij}\Lambda(y;\theta_j) + \sum_{k=N_L+1}^{N} w_{ik} P(y;\theta_k).\] 

It is reasonable to assume that $F_i(y)$ can be written as such as sum as we assume $f$ and $y$ to be analytic, which implies that $F$ is smooth. This means that a sufficiently rich basis of logistic and RBFs will approximate each $F_i$ in such a manner. In the real Koopman learning problem  we would also have a constant and weighted sum of the measurements to help approximate each $F_i$. Using these extra terms to approximate $F_i$ in our analysis makes the math more cumbersome, and so we represent $F_i$ without them. The class of representable vector fields, $F$s, is extremely broad as both logistic and RBFs are universal function approximators \cite{barron1993universal}, \cite{buhmann2003radial}.

\textit{Corollary}
\ref{cor:logApprox}:
Under the assumptions of Theorem \ref{thm:SILLconv}, the Lie derivative of a conjunctive logistic function exponentially approaches a nonlinear combination of augSILL functions, specifically,
\be
\dot{\Lambda}(y&;\theta_l) \rightarrow \sum_{i=1}^{m}\sum_{j=1}^{N_L} \alpha_{li}w_{ij}(1 - \lambda(y_i;\theta_{li})) \Lambda(y;\theta^*)  \\& +\sum_{i=1}^{m}\sum_{k=N_L+1}^{N} \alpha_{li}w_{ik}(1 - \lambda(y_i;\theta_{li})) H(y;\theta_{l},\theta_{k})
\ee exponentially as $\alpha\rightarrow\infty$.

\begin{proof}
The derivative of a conjunctive logistic function, $\Lambda(y;\theta_l)$, with respect to time is 
\be
\dot{\Lambda}(y;\theta_l) = \left(\nabla_y\Lambda(y;\theta_l)\right)^T\frac{dy}{dt} = \left(\nabla_y \Lambda(y;\theta_l)\right)^T F(y), 
\ee
where the $i^{\text{th}}$ term of the gradient of $\Lambda(y;\theta_l)$ is 
\be
\left[\nabla_y \Lambda(y;\theta_l) \right]_i &=\alpha_{li}(\lambda(y_i;\theta_{li}) - \lambda(y_i;\theta_{li})^2) \frac{\Lambda(y;\theta_l)}{\lambda(y_i;\theta_{li})} \\
&= \alpha_{li}(1 - \lambda(y_i;\theta_{li})) \Lambda(y;\theta_l).
\ee

Thus the time-derivative of $\Lambda(y_i;\theta_l)$ is
\be\label{eq:augSILL_logPrime}
\dot{\Lambda}&(y;\theta_l) = \sum_{i=1}^{m} \alpha_{li}(1 - \lambda(y_i;\theta_{li})) \Lambda(y;\theta_l) F_i(y)\\ 
&= \sum_{i=1}^{m} \alpha_{li}(1 - \lambda(y_i;\theta_{li})) \Lambda(y;\theta_{l}) (\sum_{j=1}^{N_L}w_{ij}  \Lambda(y;\theta_{j})  \\ &  + \sum_{k=N_L+1}^{N} w_{ik}P(y;\theta_{k}))\\ 
& = \sum_{i=1}^{m}\sum_{j=1}^{N_L} \alpha_{li}w_{ij}(1 - \lambda(y_i;\theta_{li})) \Lambda(y;\theta_{l})\Lambda(y;\theta_{j})  \\& +\sum_{i=1}^{m}\sum_{k=N_L+1}^{N} \alpha_{li}w_{ik}(1 - \lambda(y_i;\theta_{li})) \Lambda(y;\theta_{l})P(y;\theta_{k}).
\ee

From Theorems \ref{thm:SILLconv}, \ref{thm:augSILLconv1} and \ref{thm:augSILLconv2} we have that Eq. (\ref{eq:augSILL_logPrime}) goes to Eq. (\ref{eq:augSILL_limApproxLog}) as $\alpha\rightarrow\infty$.
$\blacksquare$\end{proof}

\textit{Corollary}
\ref{cor:rbfApprox}:
 Under the assumptions of Theorem \ref{thm:SILLconv}, the Lie derivative of a conjunctive RBF exponentially approaches a nonlinear combination of conjunctive RBFs, specifically,
 \be 
\dot{P}(y;\theta_l)& \rightarrow  \sum_{i=1}^{m}\sum_{j=1}^{N_L} \alpha_{li}w_{ij}(1 - 2\lambda(y_i;\theta_{li})) H(y;\theta_{j},\theta_{l})
\ee exponentially as $\alpha\rightarrow \infty$.
 
\begin{proof}
The derivative of a conjunctive RBF is: 
\be
\dot{P}(y;\theta_l) = \left(\nabla_yP(y;\theta_l)\right)^T\frac{dy}{dt} = \left(\nabla_y P(y;\theta_l)\right)^T F(y) 
\ee
where the $i^{\text{th}}$ term of the gradient of $P(y;\theta_l)$ is 
\be
\left[\nabla_y P(y;\theta_l) \right]_i &=\alpha_{li}\rho(y_i;\theta_{li})(1 - 2\lambda(y_i;\theta_{li})) \frac{P(y;\theta_l)}{\rho(y_i;\theta_{li})} \\
&= \alpha_{li}(1 - 2\lambda(y_i;\theta_{li})) P(y;\theta_l).
\ee

So, the time-derivative of $P(y;\theta_l)$ is
\be\label{eq:augSILL_rbfPrime}
\dot{P}&(y;\theta_l) = \sum_{i=1}^{m} \alpha_{li}(1 - 2\lambda(y_i;\theta_{li})) P(y;\theta_l) F_i(y)\\ 
&= \sum_{i=1}^{m} \alpha_{li}(1 - 2\lambda(y_i;\theta_{li})) P(y;\theta_l)(\sum_{j=1}^{N_L}w_{ij}  \Lambda(y;\theta_{j})  \\ &  + \sum_{k=N_L+1}^{N} w_{ik}P(y;\theta_{k})) \\
& = \sum_{i=1}^{m}\sum_{j=1}^{N_L} \alpha_{li}w_{ij}(1 - 2\lambda(y_i;\theta_{li})) P(y;\theta_{l})\Lambda(y;\theta_{j})  \\ +&\sum_{i=1}^{m}\sum_{k=N_L+1}^{N} \alpha_{li}w_{ik}(1 - 2\lambda(y_i;\theta_{li})) P(y;\theta_{l})P(y;\theta_{k}).
\ee

From Theorems \ref{thm:augSILLconv1}, \ref{thm:augSILLconv2} and \ref{thm:augSILLconv4} we have that Eq. (\ref{eq:augSILL_rbfPrime}) goes to Eq. (\ref{eq:augSILL_limApproxRbf}) as $\alpha\rightarrow\infty$.
$\blacksquare$\end{proof}

 \section{PDF of X(Y-Z) and Logistic and RBFs}\label{sec:AppendixPDF}
 We can compute the PDF of X(Y-Z), a term common in both logistic and RBFs. The random variables X and Z represent the scalar parameters and the random variable Y represents the scalar measurement. So, X(Y-Z) corresponds to $\alpha(y-\mu)$. We choose X, Y and Z to be identically and independently distributed (iid) as the symmetric uniform distribution: $U(-a, a), a\in\R^+$.  We then apply the law of the unconscious statistician (LOTUS) to compute the expected value of a logistic and RBF.

We compute the PDF of X(Y-Z), where the random variables X, Y and Z are independently and identically distributed as the symmetric uniform distribution: $U(-a, a), a\in\R^+$.

We start by computing the PDF of (Y-Z).  The PDF of each random variable is \be\label{eq:uniformPDF}
f_U(x) = \begin{cases}
 \frac{1}{2a} &\mbox{if }x\in [-a, a] \\
 0 &\mbox{otherwise}.
\end{cases}\ee  Since the random variables are symmetrically distributed, this is the same distribution as (Y+Z),  a symmetric triangular distribution.  The PDF is  \be\label{eq:PDFSum}
f_T(x) = \begin{cases}
 \frac{1}{2a} + \frac{x}{4a^2} &\mbox{if } x\in[-2a, 0) \\
 \frac{1}{2a} - \frac{x}{4a^2} &\mbox{if } x\in(0, 2a] \\
 0 &\mbox{otherwise}.
\end{cases}\ee 

Now we use the formula for the distribution of a product of random variables, \be g(z) = \int_{-\infty}^{\infty}f_T(x)f_U(z/x)\frac{1}{|x|}dx, \ee to compute the PDF of X(Y-Z).  So, \be g(z)&= \int_{-\infty}^{\infty}\begin{cases}
 \frac{1}{2a} + \frac{x}{4a^2} &\mbox{if } x\in[-2a, 0) \\
 \frac{1}{2a} - \frac{x}{4a^2} &\mbox{if } x\in(0, 2a] \\
 0 &\mbox{otherwise}
\end{cases}\\&\times\begin{cases}
 \frac{1}{2a|x|} &\mbox{if }\frac{z}{x}\in [-a, a] \\
 0 &\mbox{otherwise}
\end{cases}dx\\
&= \int_{0}^{2a}
 (\frac{1}{2a} - \frac{x}{4a^2} 
)\begin{cases}
 \frac{1}{2ax} &\mbox{if }\frac{z}{x}\in [-a, a] \\
 0 &\mbox{otherwise}
\end{cases} \\
&- \int_{-2a}^{0}
 (\frac{1}{2a} + \frac{x}{4a^2} 
)\begin{cases}
 \frac{1}{2ax} &\mbox{if }\frac{z}{x}\in [-a, a] \\
 0 &\mbox{otherwise.}
\end{cases}
\ee

The following equivalences hold, however they are only equal to $g(z)$ when $\frac{|z|}{a}\leq 2a$ or, alternatively written, $z\in [-2a^2, 2a^2]$.

\be
g(z) &= \frac{1}{4a^2}\int_{\frac{|z|}{a}}^{2a}
 (\frac{1}{x} - \frac{1}{2a} 
) - \frac{1}{4a^2}\int_{-2a}^{\frac{-|z|}{a}}
 (\frac{1}{x} + \frac{1}{2a} 
)\\
&= \frac{1}{4a^2}(\ln(x)-\frac{x}{2a})|_{x=\frac{|z|}{a}}^{x=2a} \\&- \frac{1}{4a^2}(\ln(-x)+\frac{x}{2a})|_{x=-2a}^{x=\frac{-|z|}{a}}\\
&= \frac{1}{4a^2}(\ln(2a) - \ln(\frac{|z|}{a})-1+\frac{|z|}{2a^2}) \\&- \frac{1}{4a^2}(\ln(\frac{|z|}{a}) - \ln(2a)-\frac{|z|}{2a^2}+1)\\
&= \frac{1}{2a^2}(\ln(\frac{2a^2}{|z|}) + \frac{|z|}{2a^2} - 1).
\ee

So, we have that our final PDF is 
\be
g(z)= \begin{cases}
 \frac{1}{2a^2}(\ln(\frac{2a^2}{|z|}) + \frac{|z|}{2a^2} - 1) & \mbox{if } z\in I\\
 0 & \mbox{else},
\end{cases}
\ee where $I$ is the interval $[-2a^2, 2a^2].$

With this PDF we use the LOTUS to compute the expected value of a randomly sampled logistic and RBF as \be 
E[\lambda] = \int_{-\infty}^{\infty}g(x)(\frac{1}{1+e^{-x}})dx
\ee and 
\be
E[\rho] = \int_{-\infty}^{\infty}g(x)(\frac{e^{-x}}{(1+e^{-x})^2})dx.
\ee

We compute the variances of these distributions using the formula: $Var[X] = E[X^2] - E[X]^2$, which involves an additional application of the LOTUS.

To the authors' knowledge, none of these integrals has a closed form solution in terms of the sampling interval parameter, $a$. Hence, we use numeric simulation to understand this relationship (see Fig. \ref{fig:expectedVals}).

\section{Dynamic Systems for Numeric Simulations}\label{sec:AppendixSystems}
The systems we test in Section \ref{sec:numerical} are: \begin{enumerate}
    \item the Van Der Pol Oscillator: \be\label{eq:vanDerPol}
    \dot x_1 &= x_2 \\
    \dot x_2&= -x_1 + c_1(1-x_1^2)x_2,
    \ee 
    \item the Duffing Oscillator: \be\label{eq:duff} 
    \dot x_1 &= x_2 \\
    \dot x_2&= -c_2 x_2 - c_3 x_1 - c_4 x_1^3,
    \ee 
    \item the Predator-Prey System: \be\label{eq:pred}
    \dot x_1 &= c_5x_1-c_6x_1x_2 \\
    \dot x_2&= c_7 x_1x_2 - c_8 x_2,
    \ee 
    \item and the Toggle Switch: \be\label{eq:togg}
    \dot x_1 &= \frac{c_9}{1+x_2^{c_{11}}}-c_{13} x_1 \\
    \dot x_2&= \frac{c_{10}}{1 + x_1^{c_{12}}}-c_{13} x_2,
    \ee 
\end{enumerate} where $c_1=1, c_2=0, c_3=-1, c_4=1, c_5=1.1, c_6=0.5, c_7=0.1, c_8=0.2, c_9=2.5, c_{10}=1.5, c_{11}=1.4, c_{12}=1.1,$ and $c_{13}=0.25$.
\end{document}